\documentclass[11pt]{article}
\usepackage{amsfonts}
\usepackage{amsmath}
\usepackage{amssymb}
\usepackage{hyperref}

\newenvironment{proof}[1][Proof]{\noindent\textbf{#1.} }{\ \rule{0.5em}{0.5em}}
\input{tcilatex}
\begin{document}

\title{\textbf{On divergences tests for composite hypotheses under composite
likelihood}}
\author{N. Mart\'{\i}n$^{1}$, L. Pardo$^{2}$ K. Zografos$^{3}$ \\
\\
$^{1}${\small Department of Statistics and O.R. II, Complutense University
of Madrid, 28003 Madrid, Spain}\\
$^{2}${\small Department of Statistics and O.R. I, Complutense University of
Madrid, 28040 Madrid, Spain}\\
$^{3}${\small Department of Mathematics, University of Ioannina, 45110
Ioannina, Greece}}
\date{}
\maketitle

\begin{abstract}
It is well-known that in some situations it is not easy to compute the
likelihood function as the datasets might be large or the model is too
complex. In that contexts composite likelihood, derived by multiplying the
likelihoods of subjects of the variables, may be useful. The extension of
the classical likelihood ratio test statistics to the framework of composite
likelihoods is used as a procedure to solve the problem of testing in the
context of composite likelihood. In this paper we introduce and study a new
family of test statistics for composite likelihood: Composite $\phi $%
-divergence test statistics for solving the problem of testing a simple null
hypothesis or a composite null hypothesis. To do that we introduce and study
the asymptotic distribution of the restricted maximum composite likelihood
estimate.
\end{abstract}

\noindent \underline{\textbf{AMS 2001 Subject Classification}}\textbf{: }


\noindent

\section{Introduction\label{Sec1}}

Hypothesis testing is a cornerstone of mathematical statistics and,
subsequently, the theory of log-likelihood ratio tests is a cornerstone in
the theory of testing statistical hypotheses, too. On the other hand,
maximum likelihood estimators play a key role in the development of
log-likelihood ratio tests. Albeit maximum likelihood estimators can be
easily obtained and they obey nice large sample properties, there are cases,
like the case of complicated probabilistic models where the maximum
likelihood estimators do not exist or they can not be obtained. In such a
case the problem is usually overcomed by the use of pseudo-likelihood
functions and the respective estimators which result by maximization of such
a function. Composite likelihood and the respective composite likelihood
estimators are an appealing case of pseudo-likelihood estimators. There is
an extensive literature composite likelihood methods in Statistics. The
history of the composite likelihood may be traced back to the
pseudo-likelihood approach of Besag (1974) for modeling spatial data. The
name of composite likelihood was given by Lindsay (1988) to refer a
likelihood type object formed by multiplying together individual component
likelihoods, each of which corresponds to a marginal or conditional event.
Composite likelihood finds applications in a variety of fields, including
genetics, spatial statistics, longitudinal analysis, multivariate modeling,
to mention a few. The special issue, with guest editors Reid, Lindsay and
Liang (2011), of the journal \textit{Statistica Sinica} is devoted to
composite likelihood methods and applications in several fields and it,
moreover, provides with an exhaustive and updated source of knowledge in the
subject. The recent papers by Reid (2013) and Cattelan and Sartori (2016)
concentrate on new developments on the composite likelihood inference.

Distance or divergence based on methods of estimation and testing are
fundamental tools and constitute a methodological part in the field of
statistical inference. The monograph by Kullback (1959) was probably the
starting point\ of usage of the so called divergence measure for testing
statistical hypotheses. The next important steps were the monographs by Read
and Cressie (1988), Vajda (1989), Pardo (2006) and Basu et al. (2011) where
distance, divergence or disparity methods were developed for estimation and
testing. Thousands of papers have been also published in this frame and many
of them have been exploited and mentioned in the above monographs. For
testing a statistical hypothesis in a parametric framework, a test-statistic
can be constructed by means of a distance or divergence measure between the
empirical model and the model which is specified by the null hypothesis. The
empirical model is the parametric model which governs the data with the
unknown parameters to be replaced by their maximum likelihood estimators.
The asymptotic normality of the maximum likelihood estimators is exploited
along with the well known delta method in order to reach the asymptotic
distribution of the respective divergence test-statistics.

The divergence test-statistics are based on considering the distance between
density functions, chosen in an appropriate way. In the statistical
situations in which we only have composite densities it seems completely
natural to define statistical procedures of testing based on divergence
measures but between the composite densities instead of the densities. This
paper is motivated by the necessity to develop divergence based on methods,
described above, for testing statistical hypotheses when the maximum
composite likelihood estimators are used instead of the classic maximum
likelihood estimators and we consider divergence measures between composite
density functions in order to get an appropriate test-statistic. In this
framework, the next section introduces the notation which will be used and
reviews composite likelihood estimators. Section \ref{Sec3} is devoted to
present a family of $\phi $-divergence test-statistics for testing simple
null hypothesis. The formulation of testing composite null hypotheses by
means of $\phi $-divergence type test-statistics is the subject of Section %
\ref{Sec5}. But in order to get the results in relation to the composite
null hypothesis it is necessary in Section \ref{Sec4} to introduce and study
the restricted maximum composite estimator as well its asymptotic
distribution and the relationship between the restricted and the
un-restricted maximum composite likelihood estimators. Section \ref{Sec6} is
devoted to present a numerical example and finally a simulation study is
carried out in Section \ref{Sec7}. The proofs of the main theoretic results
are provided in the Appendix.

\section{Composite likelihood and divergence tests\label{Sec2}}

We adopt here the notation by Joe et al. (2012) regarding composite
likelihood function and the respective maximum composite likelihood
estimators. In this regard, let $\{f(\cdot ;\boldsymbol{\theta }\mathbf{)},%
\boldsymbol{\theta }\in \Theta \subseteq 
\mathbb{R}
^{p},p\geq 1\}$ be a parametric identifiable family of distributions for an
observation $\boldsymbol{y}$, a realization of a random $m$-vector $%
\boldsymbol{Y}$. In this setting, the composite density based on $K$
different margins or conditional distributions has the form%
\begin{equation*}
\mathcal{CL}(\boldsymbol{\theta }\mathbf{,}\boldsymbol{y}\mathbf{)}%
=\tprod\limits_{k=1}^{K}f_{A_{k}}^{w_{k}}(y_{j},j\in A_{k};\boldsymbol{%
\theta })
\end{equation*}%
and the composite\ log-density based on $K$ different margins or conditional
distributions has the form%
\begin{equation*}
c\ell (\boldsymbol{\theta }\mathbf{,}\boldsymbol{y}\mathbf{)=}%
\dsum\limits_{k=1}^{K}w_{k}\ell _{A_{k}}(\boldsymbol{\theta }\mathbf{,}%
\boldsymbol{y}),
\end{equation*}%
with%
\begin{equation*}
\ell _{A_{k}}(\boldsymbol{\theta }\mathbf{,}\boldsymbol{y})=\log
f_{A_{k}}(y_{j},j\in A_{k};\boldsymbol{\theta }),
\end{equation*}%
where $\{A_{k}\}_{k=1}^{K}$ is a family of random variables\ associated
either with marginal or conditional distributions involving some $y_{j}$, $%
j\in \{1,...,m\}$ and $w_{k}$, $k=1,...,K$ are non-negative and known
weights. If the weights are all equal, then they can be ignored, actually
all the statistical procedures produce equivalent results.

Let also $\boldsymbol{y}_{1},...,\boldsymbol{y}_{n}$ be independent and
identically distributed replications of $\boldsymbol{y}$. We denote by 
\begin{equation*}
c\ell (\boldsymbol{\theta }\mathbf{,}\boldsymbol{y}_{1},...,\boldsymbol{y}%
_{n}\mathbf{)}=\dsum\limits_{i=1}^{n}c\ell (\boldsymbol{\theta }\mathbf{,}%
\boldsymbol{y}_{i}\mathbf{)}
\end{equation*}%
the composite log-likelihood function for the whole sample. In complete
accordance with the classic maximum likelihood estimator, the maximum
composite likelihood estimator $\widehat{\boldsymbol{\theta }}_{c}$ is
defined by%
\begin{equation*}
\widehat{\boldsymbol{\theta }}_{c}=\underset{\boldsymbol{\theta }\in \Theta }%
{\arg \max }\dsum\limits_{i=1}^{n}c\ell (\boldsymbol{\theta }\mathbf{,}%
\boldsymbol{y}_{i}\mathbf{)}=\underset{\boldsymbol{\theta }\in \Theta }{\arg
\max }\dsum\limits_{i=1}^{n}\dsum\limits_{k=1}^{K}w_{k}\ell _{A_{k}}(%
\boldsymbol{\theta }\mathbf{,}\boldsymbol{y}_{i}).
\end{equation*}%
It can be also obtained by the solution of the equation%
\begin{equation*}
\boldsymbol{u}(\boldsymbol{\theta }\mathbf{,}\boldsymbol{y}_{1}\mathbf{,...,}%
\boldsymbol{y}_{n}\mathbf{)}=\boldsymbol{0}_{p}\mathbf{,}
\end{equation*}%
where%
\begin{equation*}
\boldsymbol{u}(\boldsymbol{\theta }\mathbf{,}\boldsymbol{y}_{1}\mathbf{,...,}%
\boldsymbol{y}_{n}\mathbf{)}=\frac{\partial c\ell (\boldsymbol{\theta }%
\mathbf{,}\boldsymbol{y}_{1}\mathbf{,...,}\boldsymbol{y}_{n}\mathbf{)}}{%
\partial \boldsymbol{\theta }}=\dsum\limits_{i=1}^{n}\dsum%
\limits_{k=1}^{K}w_{k}\frac{\partial \ell _{A_{k}}(\boldsymbol{\theta }%
\mathbf{,}\boldsymbol{y})}{\partial \boldsymbol{\theta }},
\end{equation*}%
is the composite likelihood score function, that is the partial derivative
of the composite log-likelihood with respect to the parameter vector.

The maximum composite likelihood estimator $\widehat{\boldsymbol{\theta }}%
_{c}$ obeys asymptotic normality and in particular%
\begin{equation*}
\sqrt{n}(\widehat{\boldsymbol{\theta }}_{c}-\boldsymbol{\theta })\underset{%
n\rightarrow \infty }{\overset{\mathcal{L}}{\longrightarrow }}\mathcal{N}%
\left( \boldsymbol{0},\boldsymbol{G}_{\ast }^{-1}(\boldsymbol{\theta }%
\mathbf{)}\right) ,
\end{equation*}%
where $\boldsymbol{G}_{\ast }(\boldsymbol{\theta }\mathbf{)}$ denotes
Godambe information matrix, defined by%
\begin{equation*}
\boldsymbol{G}_{\ast }(\boldsymbol{\theta }\mathbf{)}=\boldsymbol{H}\mathbf{(%
}\boldsymbol{\theta }\mathbf{)}\boldsymbol{J}^{-1}\mathbf{(}\boldsymbol{%
\theta }\mathbf{)}\boldsymbol{H}\mathbf{(}\boldsymbol{\theta }\mathbf{),}
\end{equation*}%
with $\boldsymbol{H}\mathbf{(}\boldsymbol{\theta }\mathbf{)}$ being the
sensitivity or Hessian matrix and $\boldsymbol{J}\mathbf{(}\boldsymbol{%
\theta }\mathbf{)}$ being the variability matrix, defined, respectively, by%
\begin{align*}
\boldsymbol{H}\mathbf{(}\boldsymbol{\theta }\mathbf{)}& =E_{\boldsymbol{%
\theta }}[-\tfrac{\partial }{\partial \boldsymbol{\theta }}\boldsymbol{u}%
^{T}(\boldsymbol{\theta }\mathbf{,}\boldsymbol{Y}\mathbf{)]}, \\
\boldsymbol{J}\mathbf{(}\boldsymbol{\theta }\mathbf{)}& =Var_{\boldsymbol{%
\theta }}[\boldsymbol{u}(\boldsymbol{\theta }\mathbf{,}\boldsymbol{Y}\mathbf{%
)}]=E_{\boldsymbol{\theta }}[\boldsymbol{u}(\boldsymbol{\theta }\mathbf{,}%
\boldsymbol{Y}\mathbf{)}\boldsymbol{u}^{T}(\boldsymbol{\theta }\mathbf{,}%
\boldsymbol{Y}\mathbf{)}],
\end{align*}%
where the superscript $T$ denotes the transpose of a vector or a matrix.

The matrices $\boldsymbol{H}\mathbf{(}\boldsymbol{\theta }\mathbf{)}$ and $%
\boldsymbol{J}\mathbf{(}\boldsymbol{\theta }\mathbf{)}$ are, by definition,
nonegative definite matrices but throughout this paper both, $\boldsymbol{H}%
\mathbf{(}\boldsymbol{\theta }\mathbf{)}$ and $\boldsymbol{J}\mathbf{(}%
\boldsymbol{\theta }\mathbf{)}$, are assumed to be positive definite
matrices. Since the component score functions can be correlated, we have $%
\boldsymbol{H}\mathbf{(}\boldsymbol{\theta }\mathbf{)}\neq \boldsymbol{J}%
\mathbf{(}\boldsymbol{\theta }\mathbf{)}$. If $c\ell (\boldsymbol{\theta }%
\mathbf{,}\boldsymbol{y}\mathbf{)}$ is a true log-likelihood function then $%
\boldsymbol{H}\mathbf{(}\boldsymbol{\theta }\mathbf{)}=\boldsymbol{J}\mathbf{%
(}\boldsymbol{\theta }\mathbf{)}=\boldsymbol{I}_{F}\mathbf{(}\boldsymbol{%
\theta }\mathbf{)}$, being $\boldsymbol{I}_{F}\mathbf{(}\boldsymbol{\theta }%
\mathbf{)}$\ the Fisher information matrix of the model. Using multivariate
version of the Cauchy-Schwarz inequality we have that the matrix $%
\boldsymbol{G}_{\ast }(\boldsymbol{\theta }\mathbf{)}-\boldsymbol{I}_{F}%
\mathbf{(}\boldsymbol{\theta }\mathbf{)}$\ is non-negative definite, i.e.,
the full likelihood function is more efficient than any other composite
likelihood function (cf. Lindsay, 1988, Lemma 4A).

For two densities $p$ and $q$ associated with two $m$-dimensional random
variables respectively, Csisz\'{a}r's $\phi $-divergence between $p$ and $q$
is defined by%
\begin{equation*}
D_{\phi }(p,q)=\int_{%
\mathbb{R}
^{m}}q(\boldsymbol{y})\phi \left( \frac{p(\boldsymbol{y})}{q(\boldsymbol{y})}%
\right) d\boldsymbol{y},
\end{equation*}%
where $\phi $ is a real valued convex function, satisfying appropriate
conditions which ensure the existence of the above integral (cf., Csisz\'{a}%
r, 1963, 1967, Ali and Silvey 1963, and Pardo, 2006. Csisz\'{a}r's $\phi $%
-divergence has been axiomatically characterized and studied extensively by
Liese and Vajda (1987, 2006), Vajda (1989), and Stummer and Vajda (2010),
among many others. Particular choices of the convex functions $\phi $, lead
to important measures of divergence including\ Kullback and Leibler (1951)
divergence, R\'{e}nyi (1960) divergence and Cressie and Read (1984) $\lambda 
$-power divergence, to mention a few. Csisz\'{a}r's $\phi $-divergence can
be extended and used in testing hypotheses on more than two distributions
(cf. Zografos (1998) and references appeared therein).

In this paper we are going to consider $\phi $-divergence measures between
the composite densities $\mathcal{CL}(\boldsymbol{\theta }_{1}\mathbf{,}%
\boldsymbol{y}\mathbf{)}$ and $\mathcal{CL}(\boldsymbol{\theta }_{2}\mathbf{,%
}\boldsymbol{y}\mathbf{)}$ in order to solve different problems of testing
hypotheses. The $\phi $-divergence measure between composite densities $%
\mathcal{CL}(\boldsymbol{\theta }_{1}\mathbf{,}\boldsymbol{y}\mathbf{)}$ and 
$\mathcal{CL}(\boldsymbol{\theta }_{2}\mathbf{,}\boldsymbol{y}\mathbf{)}$
will be defined by 
\begin{equation}
D_{\phi }(\boldsymbol{\theta }_{1}\mathbf{,}\boldsymbol{\theta }_{2})=\dint_{%
\mathbb{R}
^{m}}\mathcal{CL}(\boldsymbol{\theta }_{2}\mathbf{,}\boldsymbol{y}\mathbf{)}%
\phi \left( \frac{\mathcal{CL}(\boldsymbol{\theta }_{1}\mathbf{,}\boldsymbol{%
y}\mathbf{)}}{\mathcal{CL}(\boldsymbol{\theta }_{2}\mathbf{,}\boldsymbol{y}%
\mathbf{)}}\right) d\boldsymbol{y},  \label{1}
\end{equation}%
$\phi \in $ $\Psi $, with%
\begin{equation*}
\Psi =\{\phi :\phi \text{ is strictly convex, }\phi (1)=\phi ^{\prime
}(1)=0,0\phi \left( \tfrac{0}{0}\right) =0,0\phi \left( \tfrac{u}{0}\right)
=u\underset{v\rightarrow \infty }{\lim }\tfrac{\phi (v)}{v}\}\text{.}
\end{equation*}

An important particular case is the Kullback-Leibler divergence measure
obtained from (\ref{1}) with $\phi (x)=x\log x-x+1$, i.e.%
\begin{equation*}
D_{Kullback}(\boldsymbol{\theta }_{1}\mathbf{,}\boldsymbol{\theta }%
_{2})=\dint_{%
\mathbb{R}
^{m}}\mathcal{CL}(\boldsymbol{\theta }_{1}\mathbf{,}\boldsymbol{y}\mathbf{)}%
\log \frac{\mathcal{CL}(\boldsymbol{\theta }_{1}\mathbf{,}\boldsymbol{y}%
\mathbf{)}}{\mathcal{CL}(\boldsymbol{\theta }_{2}\mathbf{,}\boldsymbol{y}%
\mathbf{)}}d\boldsymbol{y}\mathbf{.}
\end{equation*}

Based on (\ref{1}) we shall present in this paper some new test-statistics
for testing simple null hypothesis as well as composite null hypothesis. To
the best of our knowledge, it is the first time that $\phi $-divergences are
used for solving testing problems in the context of composite likelihood.
However, the Kullback-Leibler divergence has been used, in the context of
composite likelihood, by many authors in model selection, see for instance
Varin (2008).

\section{Hypothesis testing: Simple null hypothesis\label{Sec3}}

In this section we are interested in testing 
\begin{equation}
H_{0}:\boldsymbol{\theta }=\boldsymbol{\theta }_{0}\text{ versus }H_{1}:%
\boldsymbol{\theta }\neq \boldsymbol{\theta }_{0}.  \label{H1}
\end{equation}%
If we consider the $\phi $-divergence between the composite densities $%
\mathcal{CL}(\widehat{\boldsymbol{\theta }}_{c}\mathbf{,}\boldsymbol{y}%
\mathbf{)}$ and $\mathcal{CL}(\boldsymbol{\theta }_{0}\mathbf{,}\boldsymbol{y%
}\mathbf{),}$%
\begin{equation*}
D_{\phi }(\widehat{\boldsymbol{\theta }}_{c}\mathbf{,}\boldsymbol{\theta }%
_{0})=\dint_{%
\mathbb{R}
^{m}}\mathcal{CL}(\boldsymbol{\theta }_{0}\mathbf{,}\boldsymbol{y}\mathbf{)}%
\phi \left( \frac{\mathcal{CL}(\widehat{\boldsymbol{\theta }}_{c}\mathbf{,}%
\boldsymbol{y}\mathbf{)}}{\mathcal{CL}(\boldsymbol{\theta }_{0}\mathbf{,}%
\boldsymbol{y}\mathbf{)}}\right) d\boldsymbol{y}
\end{equation*}%
verifies $D_{\phi }(\widehat{\boldsymbol{\theta }}_{c}\mathbf{,}\boldsymbol{%
\theta }_{0})\geq 0$, and the equality holds if and only if $\mathcal{CL}(%
\widehat{\boldsymbol{\theta }}_{c}\mathbf{,}\boldsymbol{y}\mathbf{)}=%
\mathcal{CL}(\boldsymbol{\theta }_{0}\mathbf{,}\boldsymbol{y}\mathbf{)}$%
\textbf{.} Small values of $D_{\phi }(\widehat{\boldsymbol{\theta }}_{c}%
\mathbf{,}\boldsymbol{\theta }_{0})$ are in favour of $H_{0}:\boldsymbol{%
\theta }=\boldsymbol{\theta }_{0}$, while large values of $D_{\phi }(%
\widehat{\boldsymbol{\theta }}_{c}\mathbf{,}\boldsymbol{\theta }_{0}),$
suggest rejection of $H_{0}$. This is due to the fact that large values of $%
D_{\phi }(\widehat{\boldsymbol{\theta }}_{c}\mathbf{,}\boldsymbol{\theta }%
_{0})$ suggest that the model $\mathcal{CL}(\widehat{\boldsymbol{\theta }}%
_{c}\mathbf{,}\boldsymbol{y}\mathbf{)}$ is not very close to $\mathcal{CL}(%
\boldsymbol{\theta }_{0}\mathbf{,}\boldsymbol{y}\mathbf{)}$. Therefore, $%
H_{0}$ is rejected if $D_{\phi }(\widehat{\boldsymbol{\theta }}_{c}\mathbf{,}%
\boldsymbol{\theta }_{0})>c$, where $c$ is specified so that the
significance level of the test is $\alpha .$ In order to obtain $c$, in the
next theorem we shall obtain the asymptotic distribution of%
\begin{equation}
T_{\phi ,n}(\widehat{\boldsymbol{\theta }}_{c}\mathbf{,}\boldsymbol{\theta }%
_{0})=\frac{2n}{\phi ^{\prime \prime }(1)}D_{\phi }(\widehat{\boldsymbol{%
\theta }}_{c}\mathbf{,}\boldsymbol{\theta }_{0}),  \label{T}
\end{equation}%
which we shall refer to as composite $\phi $-divergence test-statistics for
testing simple null hypothesis.

\begin{theorem}
\label{Theorem1}Under the null hypothesis $H_{0}:\boldsymbol{\theta }=%
\boldsymbol{\theta }_{0},$%
\begin{equation*}
T_{\phi ,n}(\widehat{\boldsymbol{\theta }}_{c}\mathbf{,}\boldsymbol{\theta }%
_{0})\overset{\mathcal{L}}{\underset{n\rightarrow \infty }{\longrightarrow }}%
\sum_{i=1}^{k}\lambda _{i}Z_{i}^{2},
\end{equation*}%
where $\lambda _{i}$, $i=1,...,k$, are the eigenvalues of the matrix $%
\boldsymbol{J}\mathbf{(}\boldsymbol{\theta }_{0}\mathbf{\mathbf{)}}%
\boldsymbol{G}_{\ast }^{-1}(\boldsymbol{\theta }_{0}\mathbf{)},$%
\begin{equation*}
k=\mathrm{rank}\left( \boldsymbol{J}\mathbf{(}\boldsymbol{\theta }_{0}%
\mathbf{\mathbf{)}}\right) ,
\end{equation*}%
and $Z_{1},...Z_{r}$ are independent standard normal random variables.
\end{theorem}

\begin{proof}
Under the standard regularity assumptions of asymptoitc statistics (cf.
Serfling, 1980, p. 144 and Pardo, 2006, p. 58), we have%
\begin{equation*}
\frac{\partial D_{\phi }(\boldsymbol{\theta }\mathbf{,}\boldsymbol{\theta }%
_{0})}{\partial \theta }=\dint_{%
\mathbb{R}
^{m}}\frac{\partial \mathcal{CL}(\boldsymbol{\theta }\mathbf{,}\boldsymbol{y}%
\mathbf{)}}{\partial \boldsymbol{\theta }}\phi ^{\prime }\left( \frac{%
\mathcal{CL}(\boldsymbol{\theta }\mathbf{,}\boldsymbol{y}\mathbf{)}}{%
\mathcal{CL}(\boldsymbol{\theta }_{0}\mathbf{,}\boldsymbol{y}\mathbf{)}}%
\right) d\boldsymbol{y},
\end{equation*}%
therefore 
\begin{equation*}
\left. \frac{\partial D_{\phi }(\boldsymbol{\theta }\mathbf{,}\boldsymbol{%
\theta }_{0})}{\partial \theta }\right\vert _{\boldsymbol{\theta }=%
\boldsymbol{\theta }_{0}}=\phi ^{\prime }\left( 1\right) \dint_{%
\mathbb{R}
^{m}}\left. \frac{\partial \mathcal{CL}(\boldsymbol{\theta }\mathbf{,}%
\boldsymbol{y}\mathbf{)}}{\partial \boldsymbol{\theta }}\right\vert _{%
\boldsymbol{\theta }=\boldsymbol{\theta }_{0}}d\boldsymbol{y}=\boldsymbol{0}%
_{p}.
\end{equation*}%
On the other hand,%
\begin{align*}
& \frac{\partial ^{2}D_{\phi }(\boldsymbol{\theta }\mathbf{,}\boldsymbol{%
\theta }_{0})}{\partial \boldsymbol{\theta }\partial \boldsymbol{\theta }^{T}%
} \\
& =\dint_{%
\mathbb{R}
^{m}}\frac{\partial ^{2}\mathcal{CL}(\boldsymbol{\theta }\mathbf{,}%
\boldsymbol{y}\mathbf{)}}{\partial \boldsymbol{\theta }\partial \boldsymbol{%
\theta }^{T}}\phi ^{\prime }\left( \frac{\mathcal{CL}(\boldsymbol{\theta }%
\mathbf{,}\boldsymbol{y}\mathbf{)}}{\mathcal{CL}(\boldsymbol{\theta }_{0}%
\mathbf{,}\boldsymbol{y}\mathbf{)}}\right) d\boldsymbol{y}\mathbf{+}\dint_{%
\mathbb{R}
^{m}}\frac{\partial \mathcal{CL}(\boldsymbol{\theta }\mathbf{,}\boldsymbol{y}%
\mathbf{)}}{\partial \boldsymbol{\theta }}\frac{\partial \mathcal{CL}(%
\boldsymbol{\theta }\mathbf{,}\boldsymbol{y}\mathbf{)}}{\partial \boldsymbol{%
\theta }^{T}}\frac{1}{\mathcal{CL}(\boldsymbol{\theta }_{0}\mathbf{,}%
\boldsymbol{y}\mathbf{)}}\phi ^{\prime \prime }\left( \frac{\mathcal{CL}(%
\boldsymbol{\theta }\mathbf{,}\boldsymbol{y}\mathbf{)}}{\mathcal{CL}(%
\boldsymbol{\theta }_{0}\mathbf{,}\boldsymbol{y}\mathbf{)}}\right) d%
\boldsymbol{y}
\end{align*}%
and 
\begin{equation*}
\left. \frac{\partial ^{2}D_{\phi }(\boldsymbol{\theta }\mathbf{,}%
\boldsymbol{\theta }_{0})}{\partial \boldsymbol{\theta }\partial \boldsymbol{%
\theta }^{T}}\right\vert _{\boldsymbol{\theta }=\boldsymbol{\theta }%
_{0}}=\phi ^{\prime \prime }\left( 1\right) \dint_{%
\mathbb{R}
^{m}}\left. \frac{\partial c\ell (\boldsymbol{\theta }\mathbf{,}\boldsymbol{y%
}\mathbf{)}}{\partial \boldsymbol{\theta }}\frac{\partial c\ell (\boldsymbol{%
\theta }\mathbf{,}\boldsymbol{y}\mathbf{)}}{\partial \boldsymbol{\theta }^{T}%
}\right\vert _{\boldsymbol{\theta }=\boldsymbol{\theta }_{0}}\mathcal{CL}(%
\boldsymbol{\theta }_{0}\mathbf{,}\boldsymbol{y}\mathbf{)}d\boldsymbol{y}%
=\phi ^{\prime \prime }\left( 1\right) \boldsymbol{J}\mathbf{(}\boldsymbol{%
\theta }_{0}\mathbf{\mathbf{).}}
\end{equation*}%
Then, from%
\begin{equation*}
D_{\phi }(\widehat{\boldsymbol{\theta }}_{c}\mathbf{,}\boldsymbol{\theta }%
_{0})=\frac{\phi ^{\prime \prime }\left( 1\right) }{2}(\widehat{\boldsymbol{%
\theta }}_{c}-\boldsymbol{\theta }_{0})^{T}\boldsymbol{J}\mathbf{(}%
\boldsymbol{\theta }_{0}\mathbf{\mathbf{)(}}\widehat{\boldsymbol{\theta }}%
_{c}-\boldsymbol{\theta }_{0})+o\left( n^{-1/2}\right)
\end{equation*}%
the desired result is obtained. The value of $k$ comes from 
\begin{equation*}
k=\mathrm{rank}\left( \boldsymbol{G}_{\ast }^{-1}(\boldsymbol{\theta }_{0}%
\mathbf{)}\boldsymbol{J}^{T}\mathbf{(}\boldsymbol{\theta }_{0}\mathbf{%
\mathbf{)}}\boldsymbol{G}_{\ast }^{-1}(\boldsymbol{\theta }_{0}\mathbf{)}%
\right) =\mathrm{rank}(\boldsymbol{J}\mathbf{(}\boldsymbol{\theta }_{0}%
\mathbf{\mathbf{))}}.
\end{equation*}
\end{proof}

\begin{remark}
Based on the previous Theorem we shall reject the null hypothesis $H_{0}:%
\boldsymbol{\theta }=\boldsymbol{\theta }_{0}$ if $T_{\phi ,n}(\widehat{%
\boldsymbol{\theta }}_{c}\mathbf{,}\boldsymbol{\theta }_{0})>c_{\alpha }$,
where $c_{\alpha }$ is the quantile of order $1-\alpha $ of the asymptotic
distribution of $T_{\phi ,n}(\widehat{\boldsymbol{\theta }}_{c}\mathbf{,}%
\boldsymbol{\theta }_{0})$ given in (\ref{T}). The value of $k$ is usually $%
p $, since the components of $\boldsymbol{\theta }$\ are assumed to be
non-redundant.
\end{remark}

In most cases, the power function of this testing procedure can not be
calculated explicitly. In the following theorem we present a useful
asymptotic result for approximating the power function.

\begin{theorem}
\label{Theorem2}Let $\boldsymbol{\theta }^{\ast }$ be the true parameter,
with $\boldsymbol{\theta }^{\ast }\neq \boldsymbol{\theta }_{0}$\textbf{.}
Then it holds%
\begin{equation*}
\sqrt{n}\left( D_{\phi }(\widehat{\boldsymbol{\theta }}_{c}\mathbf{,}%
\boldsymbol{\theta }_{0})-D_{\phi }(\boldsymbol{\theta }^{\ast }\mathbf{,}%
\boldsymbol{\theta }_{0})\right) \overset{\mathcal{L}}{\underset{%
n\rightarrow \infty }{\longrightarrow }}\mathcal{N}\left( 0,\sigma _{\phi
}^{2}\left( \boldsymbol{\theta }^{\ast }\right) \right) ,
\end{equation*}%
where 
\begin{equation*}
\sigma _{\phi }^{2}\left( \boldsymbol{\theta }^{\ast }\right) =\boldsymbol{q}%
^{T}\boldsymbol{G}_{\ast }^{-1}(\boldsymbol{\theta }_{0}\mathbf{)}%
\boldsymbol{q}
\end{equation*}%
and $\boldsymbol{q}=\left( q_{1},...,q_{p}\right) ^{T}$ with $q_{j}=\left. 
\frac{\partial D_{\phi }(\boldsymbol{\theta }\mathbf{,}\boldsymbol{\theta }%
_{0})}{\partial \theta _{j}}\right\vert _{\boldsymbol{\theta }=\boldsymbol{%
\theta }^{\ast }},$ $j=1,...,p$.
\end{theorem}

\begin{proof}
A first order Taylor expansion gives 
\begin{equation*}
D_{\phi }(\widehat{\boldsymbol{\theta }}_{c}\mathbf{,}\boldsymbol{\theta }%
_{0})=D_{\phi }(\boldsymbol{\theta }^{\ast }\mathbf{,}\boldsymbol{\theta }%
_{0})+\boldsymbol{q}^{T}(\widehat{\boldsymbol{\theta }}_{c}-\boldsymbol{%
\theta }^{\ast })+o(\left\Vert \widehat{\boldsymbol{\theta }}_{c}-%
\boldsymbol{\theta }^{\ast }\right\Vert ).
\end{equation*}%
But 
\begin{equation*}
\sqrt{n}(\widehat{\boldsymbol{\theta }}_{c}-\boldsymbol{\theta })\underset{%
n\rightarrow \infty }{\overset{\mathcal{L}}{\longrightarrow }}\mathcal{N}%
\left( \boldsymbol{0},\boldsymbol{G}_{\ast }^{-1}(\boldsymbol{\theta }%
\mathbf{)}\right)
\end{equation*}%
and $\sqrt{n}o(\left\Vert \widehat{\boldsymbol{\theta }}_{c}-\boldsymbol{%
\theta }^{\ast }\right\Vert )=o_{p}(1).$ Now the result follows.
\end{proof}

\begin{remark}
From Theorem \ref{Theorem2}, a first approximation to the power function, at 
$\boldsymbol{\theta }^{\ast }\neq \boldsymbol{\theta }_{0}$, is given by%
\begin{equation*}
\beta _{n,\phi }\left( \boldsymbol{\theta }^{\ast }\right) =1-\Phi \left( 
\frac{\sqrt{n}}{\sigma _{\phi }\left( \boldsymbol{\theta }^{\ast }\right) }%
\left( \frac{\phi ^{\prime \prime }(1)c_{\alpha }}{2n}-D_{\phi }(\boldsymbol{%
\theta }^{\ast }\mathbf{,}\boldsymbol{\theta }_{0})\right) \right)
\end{equation*}%
where $\Phi $ is the standard normal distribution function. If some $%
\boldsymbol{\theta }^{\ast }\neq \boldsymbol{\theta }_{0}$ is the true
parameter, the probability of rejecting $\boldsymbol{\theta }_{0}$ with the
rejection rule $T_{\phi ,n}(\widehat{\boldsymbol{\theta }}_{c}\mathbf{,}%
\boldsymbol{\theta }_{0})>c_{\alpha }$, for fixed significance level $\alpha 
$, tends to one as $n\rightarrow \infty .$ Hence, the test is consistent in
Fraser's sense.
\end{remark}

\section{Restricted maximum composite likelihood estimator\label{Sec4}}

In some common situations such as the problem of testing composite null
hypotheses, it is necessary to get the maximum composite likelihood
estimator which is restricted by some restrictions of the type 
\begin{equation}
\boldsymbol{g}(\boldsymbol{\theta })=\boldsymbol{0}_{r},  \label{res}
\end{equation}%
where $\boldsymbol{g}$ is a function such that $\boldsymbol{g}:\Theta
\subseteq 
\mathbb{R}
^{p}\longrightarrow 
\mathbb{R}
^{r}$, $r$ is an integer, with $r<p$ and $\boldsymbol{0}_{r}$ denotes the
null vector of dimension $r$. The function $\boldsymbol{g}$ is a vector
valued function such that the $p\times r$ matrix 
\begin{equation}
\boldsymbol{G}(\boldsymbol{\theta })=\frac{\partial \boldsymbol{g}^{T}(%
\boldsymbol{\theta })}{\partial \boldsymbol{\theta }}  \label{G}
\end{equation}%
exists and is continuous in $\boldsymbol{\theta }$ with $\mathrm{rank}(%
\boldsymbol{G}(\boldsymbol{\theta }))=r$. The restricted maximum composite
likelihood estimator of $\boldsymbol{\theta }$ is defined by%
\begin{equation*}
\widetilde{\boldsymbol{\theta }}_{rc}=\underset{\boldsymbol{\theta }\in
\Theta ,\boldsymbol{g}(\boldsymbol{\theta })=\boldsymbol{0}_{r}}{\arg \max }%
\dsum\limits_{i=1}^{n}c\ell (\boldsymbol{\theta }\mathbf{,}\boldsymbol{y}_{i}%
\mathbf{)}=\underset{\boldsymbol{\theta }\in \Theta ,\boldsymbol{g}(%
\boldsymbol{\theta })=\boldsymbol{0}_{r}}{\arg \max }\dsum\limits_{i=1}^{n}%
\dsum\limits_{k=1}^{K}w_{k}\ell _{A_{k}}(\boldsymbol{\theta }\mathbf{,}%
\boldsymbol{y}_{i}).
\end{equation*}%
and is obtained by the solution of the restricted likelihood equations%
\begin{align*}
\dsum\limits_{i=1}^{n}\frac{\partial }{\partial \boldsymbol{\theta }}c\ell (%
\boldsymbol{\theta }\mathbf{,}\boldsymbol{y}_{i}\mathbf{)}+\boldsymbol{G}(%
\boldsymbol{\theta })\mathbf{\lambda }& =\boldsymbol{0}_{p}\mathbf{,} \\
\boldsymbol{g}(\boldsymbol{\theta })& =\boldsymbol{0}_{r}\mathbf{,}
\end{align*}%
where $\mathbf{\lambda \in }%
\mathbb{R}
^{r}$ is a vector of Lagrange multipliers.

In this section we shall get the asymptotic distribution of the restricted
maximum composite likelihood estimator. Consider a random sample $%
\boldsymbol{y}_{1},...,\boldsymbol{y}_{n}$ from the parametric model $%
f(\cdot ;\boldsymbol{\theta }\mathbf{)},\boldsymbol{\theta }\in \Theta
\subseteq 
\mathbb{R}
^{p},p\geq 1$, and let $\widehat{\boldsymbol{\theta }}_{c}$ and $\widetilde{%
\boldsymbol{\theta }}_{rc}$ be the unrestricted and the restricted maximum
composite likelihood estimators of $\boldsymbol{\theta }$. The following
result derives the asymptotic distribution of $\widetilde{\boldsymbol{\theta 
}}_{rc}$.

\begin{theorem}
\label{Theorem3} Under the constraints $\boldsymbol{g}(\boldsymbol{\theta })=%
\boldsymbol{0}_{r}$ the\ restricted maximum composite likelihood estimator
obeys asymptotic normality in the sense 
\begin{equation*}
\sqrt{n}(\widetilde{\boldsymbol{\theta }}_{rc}-\boldsymbol{\theta })\overset{%
\mathcal{L}}{\underset{n\rightarrow \infty }{\longrightarrow }}\mathcal{N(}%
\boldsymbol{0}_{p},\widetilde{\boldsymbol{\Sigma }}_{rc}),
\end{equation*}%
with%
\begin{align*}
\widetilde{\boldsymbol{\Sigma }}_{rc}& =\boldsymbol{P}(\boldsymbol{\theta }%
\mathbf{)}\boldsymbol{J}\mathbf{(}\boldsymbol{\theta }\mathbf{)}\boldsymbol{P%
}^{T}\mathbf{(}\boldsymbol{\theta }\mathbf{),} \\
\boldsymbol{P}(\boldsymbol{\theta }\mathbf{)}& =\boldsymbol{H}^{-1}(%
\boldsymbol{\theta })+\boldsymbol{Q}(\boldsymbol{\theta }\mathbf{)}%
\boldsymbol{G}^{T}\mathbf{(}\boldsymbol{\theta }\mathbf{)}\boldsymbol{H}%
^{-1}(\boldsymbol{\theta }), \\
\boldsymbol{Q}(\boldsymbol{\theta }\mathbf{)}& =\mathbf{-}\boldsymbol{H}%
^{-1}(\boldsymbol{\theta })\boldsymbol{G}\mathbf{(}\boldsymbol{\theta }%
\mathbf{)}\left[ \boldsymbol{G}^{T}\mathbf{(}\boldsymbol{\theta }\mathbf{)}%
\boldsymbol{H}^{-1}\mathbf{(}\boldsymbol{\theta }\mathbf{)}\boldsymbol{G}%
\mathbf{(}\boldsymbol{\theta }\mathbf{)}\right] ^{-1}.
\end{align*}
\end{theorem}

The proof of the Theorem is outlined in Section \ref{ApA}\ of Appendix.

The lemma that follows formulates the relationship between the maximum
composite and the restricted maximum composite likelihood estimators $%
\widehat{\boldsymbol{\theta }}_{c}$ and $\widetilde{\boldsymbol{\theta }}%
_{rc}$ respectively.

\begin{lemma}
\label{lemma2} The estimators of $\boldsymbol{\theta }$, $\widehat{%
\boldsymbol{\theta }}_{c}$ and $\widetilde{\boldsymbol{\theta }}_{rc}$,
satisfy%
\begin{equation*}
\sqrt{n}(\widetilde{\boldsymbol{\theta }}_{rc}-\boldsymbol{\theta })=\left( 
\boldsymbol{I}_{p}+\boldsymbol{Q}(\boldsymbol{\theta })\boldsymbol{G}^{T}(%
\boldsymbol{\theta })\right) \sqrt{n}(\widehat{\boldsymbol{\theta }}_{c}-%
\boldsymbol{\theta })+o_{P}(1).
\end{equation*}
\end{lemma}

The proof of the lemma is given in Section \ref{ApB} of Appendix.

\section{Composite null hypothesis\label{Sec5}}

Following Basu et al. (2015), consider the null hypothesis 
\begin{equation*}
H_{0}:\boldsymbol{\theta }\in \Theta _{0}\text{ \ against \ }H_{0}:%
\boldsymbol{\theta }\notin \Theta _{0},
\end{equation*}%
which restricts the parameter $\boldsymbol{\theta }$ to a subset $\Theta
_{0} $ of $\Theta \subseteq 
\mathbb{R}
^{p},p\geq 1$. Based on Sen and Singer (1993, p. 239), we shall assume that
the composite null hypothesis $H_{0}:\boldsymbol{\theta }\in \Theta _{0}$
can be equivalently formulated in the form%
\begin{equation}
H_{0}:\boldsymbol{g}(\boldsymbol{\theta })=\boldsymbol{0}_{r}.  \label{H}
\end{equation}

For testing the composite null hypothesis (\ref{H}) on the basis of a random
sample $\boldsymbol{y}_{1},...,\boldsymbol{y}_{n}$ from the parametric model 
$f(\cdot ;\boldsymbol{\theta }\mathbf{)}$, $\boldsymbol{\theta }\in \Theta
\subseteq 
\mathbb{R}
^{p}$, $p\geq 1$, there are well-known procedures to be applied. The
likelihood ratio test-statistic, the Wald and Rao statistics are used in
this direction. Test-statistics based on divergences or disparities, as they
have been described and mentioned above, constitute an appealing procedure
for testing this hypothesis. Moreover, there are composite likelihood
methods analog to the likelihood ratio test or the Wald test. However, there
are not composite likelihood versions of the tests based on divergence
measures, to the best of our knowledge. So, our aim in this Section is to
develop test-statistics for testing (\ref{H}), on the basis of divergence
measures and in the composite likelihood framework. The $\phi $-divergence
between the composite densities $\mathcal{CL}(\widehat{\boldsymbol{\theta }}%
_{c}\mathbf{,}\boldsymbol{y}\mathbf{)}$ and $\mathcal{CL}(\widetilde{%
\boldsymbol{\theta }}_{rc}\mathbf{,}\boldsymbol{y}\mathbf{),}$ is given by 
\begin{equation*}
D_{\phi }(\widehat{\boldsymbol{\theta }}_{c}\mathbf{,}\widetilde{\boldsymbol{%
\theta }}_{rc})=\tint_{%
\mathbb{R}
^{m}}\mathcal{CL}(\widetilde{\boldsymbol{\theta }}_{rc}\mathbf{,}\boldsymbol{%
y}\mathbf{)}\phi \left( \frac{\mathcal{CL}(\widehat{\boldsymbol{\theta }}_{c}%
\mathbf{,}\boldsymbol{y}\mathbf{)}}{\mathcal{CL}(\widetilde{\boldsymbol{%
\theta }}_{rc}\mathbf{,}\boldsymbol{y}\mathbf{)}}\right) d\boldsymbol{y.}
\end{equation*}

Based on the property $D_{\phi }(\widehat{\boldsymbol{\theta }}_{c}\mathbf{,}%
\widetilde{\boldsymbol{\theta }}_{rc})\geq 0$, with equality, if and only if 
$\mathcal{CL}(\widetilde{\boldsymbol{\theta }}_{rc}\mathbf{,}\boldsymbol{y}%
\mathbf{)}=\mathcal{CL}(\widehat{\boldsymbol{\theta }}_{c}\mathbf{,}%
\boldsymbol{y}\mathbf{)}$, small values of $D_{\phi }(\widehat{\boldsymbol{%
\theta }}_{c}\mathbf{,}\widetilde{\boldsymbol{\theta }}_{rc})$ are in favour
of (\ref{H}), while large values of $D_{\phi }(\widehat{\boldsymbol{\theta }}%
_{c}\mathbf{,}\widetilde{\boldsymbol{\theta }}_{rc})$ suggest that the
composite densities $\mathcal{CL}(\widetilde{\boldsymbol{\theta }}_{rc}%
\mathbf{,}\boldsymbol{y}\mathbf{)}$ and $\mathcal{CL}(\widehat{\boldsymbol{%
\theta }}_{c}\mathbf{,}\boldsymbol{y}\mathbf{)}$ are not the same and the
same is expected for the respective theoretic models $f_{\theta }$ with $%
\boldsymbol{\theta }\in \Theta $ and $f_{\theta }$ with $\boldsymbol{\theta }%
\in \Theta _{0}$. So, small values of $D_{\phi }(\widehat{\boldsymbol{\theta 
}}_{c}\mathbf{,}\widetilde{\boldsymbol{\theta }}_{rc})$ are in favor of (\ref%
{H}) while large values of $D_{\phi }(\widehat{\boldsymbol{\theta }}_{c}%
\mathbf{,}\widetilde{\boldsymbol{\theta }}_{rc})$ suggest the rejection of $%
H_{0}$. Given the asymptotic normality of the maximum composite likelihood
estimator $\widehat{\boldsymbol{\theta }}_{c}$, the asymptotic normality of
the respective restricted estimator $\widetilde{\boldsymbol{\theta }}_{rc}$
should be verified. The asymptotic normality of $\widetilde{\boldsymbol{%
\theta }}_{rc}$ and the investigation of the asymptotic distribution of the
test-statistic $D_{\phi }(\widehat{\boldsymbol{\theta }}_{c}\mathbf{,}%
\widetilde{\boldsymbol{\theta }}_{rc})$ is the subject of the next section.

Based on Theorem \ref{Theorem3} and Lemma \ref{lemma2}, the composite
likelihood $\phi $-divergence test-statistic is introduced in the next
theorem and its asymptotic distribution is derived under the composite null
hypothesis (\ref{H}). The standard regularity assumptions of asymptotic
statistic are assumed to be valid (cf. Serfling, 1980, p. 144 and Pardo,
2006, p. 58).

\begin{theorem}
\label{Theorem4} Under the composite null hypothesis (\ref{H}), 
\begin{equation*}
T_{\phi ,n}(\widehat{\boldsymbol{\theta }}_{c}\mathbf{,}\widetilde{%
\boldsymbol{\theta }}_{rc})=\frac{2n}{\phi ^{\prime \prime }(1)}D_{\phi }(%
\widehat{\boldsymbol{\theta }}_{c}\mathbf{,}\widetilde{\boldsymbol{\theta }}%
_{rc})\overset{\mathcal{L}}{\underset{n\rightarrow \infty }{\longrightarrow }%
}\sum_{i=1}^{k}\beta _{i}Z_{i}^{2},
\end{equation*}%
where $\beta _{i}$, $i=1,...,k$, are the eigenvalues of the matrix 
\begin{equation*}
\boldsymbol{J}(\boldsymbol{\theta }_{0})\boldsymbol{G}\mathbf{(}\boldsymbol{%
\theta }\mathbf{)}\boldsymbol{Q}^{T}\mathbf{(}\boldsymbol{\theta }\mathbf{%
\mathbf{)}}\boldsymbol{G}_{\ast }^{-1}(\boldsymbol{\theta }\mathbf{)}%
\boldsymbol{Q}(\boldsymbol{\theta }\mathbf{)}\boldsymbol{G}^{T}\mathbf{(}%
\boldsymbol{\theta }\mathbf{)},
\end{equation*}%
\begin{equation*}
k=\mathrm{rank}\left( \boldsymbol{G}\mathbf{(}\boldsymbol{\theta }\mathbf{)}%
\boldsymbol{Q}^{T}\mathbf{(}\boldsymbol{\theta }\mathbf{\mathbf{)}}%
\boldsymbol{G}_{\ast }^{-1}(\boldsymbol{\theta }\mathbf{)}\boldsymbol{Q}(%
\boldsymbol{\theta }\mathbf{)}\boldsymbol{G}^{T}\mathbf{(}\boldsymbol{\theta 
}\mathbf{)}\boldsymbol{J}(\boldsymbol{\theta }_{0})\boldsymbol{G}\mathbf{(}%
\boldsymbol{\theta }\mathbf{)}\boldsymbol{Q}^{T}\mathbf{(}\boldsymbol{\theta 
}\mathbf{\mathbf{)}}\boldsymbol{G}_{\ast }^{-1}(\boldsymbol{\theta }\mathbf{)%
}\boldsymbol{Q}(\boldsymbol{\theta }\mathbf{)}\boldsymbol{G}^{T}\mathbf{(}%
\boldsymbol{\theta }\mathbf{)}\right) ,
\end{equation*}%
and $Z_{1},...Z_{r}$ are independent standard normal random variables.
\end{theorem}

The proof of this Theorem is presented in Appendix C. In the following we
refer $T_{\phi ,n}(\widehat{\boldsymbol{\theta }}_{c}\mathbf{,}\widetilde{%
\boldsymbol{\theta }}_{rc})$ by composite $\phi $-divergence test-statistics
for testing composite null hypothesis.

\begin{remark}
For the testing problem considered in this Section it is perhaps well-known
the composite likelihood ratio test but it was not possible for us to find
it in the statistical literature. This test will be used in Section 4 and
this is the reason to develop the said test in the present remark. \newline
We shall denote%
\begin{equation*}
c\ell \left( \boldsymbol{\theta }\right) =\tsum\limits_{i=1}^{n}c\ell \left( 
\boldsymbol{\theta },\boldsymbol{y}_{i}\right)
\end{equation*}%
The composite likelihood ratio test for testing the composite null
hypothesis (\ref{H}), considered in this paper, is defined by%
\begin{equation*}
\lambda _{n}(\widehat{\boldsymbol{\theta }}_{c},\widetilde{\boldsymbol{%
\theta }}_{rc})=2\left( c\ell (\widehat{\boldsymbol{\theta }}_{c})-c\ell (%
\widetilde{\boldsymbol{\theta }}_{rc})\right) .
\end{equation*}%
A second order Taylor expansion gives%
\begin{align*}
& c\ell (\widetilde{\boldsymbol{\theta }}_{rc},\boldsymbol{y}_{i})-c\ell (%
\widehat{\boldsymbol{\theta }}_{c},\boldsymbol{y}_{i}) \\
& =\left. \frac{\partial c\ell \left( \boldsymbol{\theta },\boldsymbol{y}%
_{i}\right) }{\partial \boldsymbol{\theta }}\right\vert _{\boldsymbol{\theta 
}=\widehat{\boldsymbol{\theta }}_{c}}(\widehat{\boldsymbol{\theta }}_{c}-%
\widetilde{\boldsymbol{\theta }}_{rc})+\frac{1}{2}(\widehat{\boldsymbol{%
\theta }}_{c}-\widetilde{\boldsymbol{\theta }}_{rc})^{T}\left. \frac{%
\partial ^{2}c\ell \left( \boldsymbol{\theta },\boldsymbol{y}_{i}\right) }{%
\partial \boldsymbol{\theta }\partial \boldsymbol{\theta }^{T}}\right\vert _{%
\boldsymbol{\theta }=\widehat{\boldsymbol{\theta }}_{c}}(\widehat{%
\boldsymbol{\theta }}_{c}-\widetilde{\boldsymbol{\theta }}_{rc})+o_{P}(1).
\end{align*}%
But,%
\begin{equation*}
\left. \frac{\partial c\ell \left( \boldsymbol{\theta },\boldsymbol{y}%
_{i}\right) }{\partial \boldsymbol{\theta }}\right\vert _{\boldsymbol{\theta 
}=\widehat{\boldsymbol{\theta }}_{c}}=\boldsymbol{0}_{p}\text{\quad and\quad 
}\frac{1}{n}\tsum\limits_{i=1}^{n}\left. \frac{\partial ^{2}c\ell \left( 
\boldsymbol{\theta },\boldsymbol{y}_{i}\right) }{\partial \boldsymbol{\theta 
}\partial \boldsymbol{\theta }^{T}}\right\vert _{\boldsymbol{\theta }=%
\widehat{\boldsymbol{\theta }}_{c}}\overset{a.s.}{\underset{n\rightarrow
\infty }{\longrightarrow }}-\boldsymbol{H}(\boldsymbol{\theta }_{0}),
\end{equation*}%
and therefore, 
\begin{equation*}
2(c\ell (\widehat{\boldsymbol{\theta }}_{c})-c\ell (\widetilde{\boldsymbol{%
\theta }}_{rc}))=\sqrt{n}(\widehat{\boldsymbol{\theta }}_{c}-\widetilde{%
\boldsymbol{\theta }}_{rc})^{T}\boldsymbol{H}(\boldsymbol{\theta }_{0})\sqrt{%
n}(\widehat{\boldsymbol{\theta }}_{c}-\widetilde{\boldsymbol{\theta }}%
_{rc})+o_{P}(1)
\end{equation*}%
which yields%
\begin{equation}
\lambda _{n}(\widehat{\boldsymbol{\theta }}_{c},\widetilde{\boldsymbol{%
\theta }}_{rc})=2(c\ell (\widehat{\boldsymbol{\theta }}_{c})-c\ell (%
\widetilde{\boldsymbol{\theta }}_{rc}))\overset{\mathcal{L}}{\underset{%
n\rightarrow \infty }{\longrightarrow }}\dsum\limits_{i=1}^{\ell }\gamma
_{i}Z_{i}^{2},  \label{R1}
\end{equation}%
where $\gamma _{i},$ $i=1,...,\ell $ are the non null eigenvalues of the
matrix%
\begin{equation*}
\boldsymbol{H}(\boldsymbol{\theta }_{0})\boldsymbol{G}(\boldsymbol{\theta })%
\boldsymbol{Q}^{T}(\boldsymbol{\theta })\boldsymbol{G}_{\ast }^{-1}(%
\boldsymbol{\theta })\boldsymbol{Q}(\boldsymbol{\theta })\boldsymbol{G}^{T}(%
\boldsymbol{\theta }),
\end{equation*}%
with%
\begin{equation*}
\ell =\mathrm{rank}\left( \boldsymbol{G}(\boldsymbol{\theta })\boldsymbol{Q}%
^{T}(\boldsymbol{\theta })\boldsymbol{G}_{\ast }^{-1}(\boldsymbol{\theta })%
\boldsymbol{Q}(\boldsymbol{\theta })\boldsymbol{G}(\boldsymbol{\theta })^{T}%
\boldsymbol{H}(\boldsymbol{\theta })\boldsymbol{G}(\boldsymbol{\theta })%
\boldsymbol{Q}^{T}(\boldsymbol{\theta })\boldsymbol{G}_{\ast }^{-1}(%
\boldsymbol{\theta })\boldsymbol{Q}(\boldsymbol{\theta })\boldsymbol{G}^{T}(%
\boldsymbol{\theta })\right) ,
\end{equation*}%
and $Z_{i},$ $i=1,...,\ell $ are independent standard normal random
variables.
\end{remark}

\begin{remark}
In order to avoid the problem of getting percentiles or probabilities from
the distribution of linear combinations of chi-squares we are going to
present some adjusted composite likelihood $\phi $-divergence
test-statistics.

Following Corollary 1 of Rao and Scott (1981) one can use the statistic 
\begin{equation*}
^{1}T_{\phi ,n}(\widehat{\boldsymbol{\theta }}_{c}\mathbf{,}\widetilde{%
\boldsymbol{\theta }}_{rc})=\frac{T_{\phi ,n}(\widehat{\boldsymbol{\theta }}%
_{c}\mathbf{,}\widetilde{\boldsymbol{\theta }}_{rc})}{\lambda _{\max }}\leq {%
\sum\limits_{i=1}^{r}}Z_{i}^{2},
\end{equation*}%
where $\lambda _{\max }=\max \left( \beta _{1},...,\beta _{r}\right) $. As ${%
\sum\limits_{i=1}^{r}}Z_{i}^{2}\sim \chi _{r}^{2}$, a strategy that rejects
the null hypothesis $H_{0}:\boldsymbol{g}(\boldsymbol{\theta })=\boldsymbol{0%
}_{r}$ for $^{1}T_{\phi ,n}(\widehat{\boldsymbol{\theta }}_{c}\mathbf{,}%
\widetilde{\boldsymbol{\theta }}_{rc})>\chi _{r,1-\alpha }^{2}$ produces an
asymptotically conservative test at $\alpha $ nominal level , where $\chi
_{r,1-\alpha }^{2}$ is the quantile of order $1-\alpha $ for $\chi _{r}^{2}$.

Another approximation to the asymptotic tail probabilities of $T_{\phi ,n}(%
\widehat{\boldsymbol{\theta }}_{c}\mathbf{,}\widetilde{\boldsymbol{\theta }}%
_{rc})$ can be obtained through the modification 
\begin{equation*}
^{2}T_{\phi ,n}(\widehat{\boldsymbol{\theta }}_{c}\mathbf{,}\widetilde{%
\boldsymbol{\theta }}_{rc})=\frac{T_{\phi ,n}(\widehat{\boldsymbol{\theta }}%
_{c}\mathbf{,}\widetilde{\boldsymbol{\theta }}_{rc})}{\overline{\lambda }},
\end{equation*}%
where $\bar{\lambda}=\frac{1}{r}\sum_{i=1}^{r}\beta _{i}$ (see
Satterthwaite, 1946), considered approximated by a chi-squared distribution
with $r$ degrees of freedom. In this case we can observe that%
\begin{align*}
E\left[ ^{2}T_{\phi ,n}(\widehat{\boldsymbol{\theta }}_{c}\mathbf{,}%
\widetilde{\boldsymbol{\theta }}_{rc})\right] & =r=E\left[ \chi _{r}^{2}%
\right] , \\
Var\left[ ^{2}T_{\phi ,n}(\widehat{\boldsymbol{\theta }}_{c}\mathbf{,}%
\widetilde{\boldsymbol{\theta }}_{rc})\right] & =\frac{2\sum_{i=1}^{r}\beta
_{i}^{2}}{\overline{\lambda }^{2}}=2r+2{\sum\limits_{i=1}^{r}}\frac{\left(
\beta _{i}-\overline{\lambda }\right) ^{2}}{\overline{\lambda }^{2}}>2r=Var%
\left[ \chi _{r}^{2}\right] .
\end{align*}%
If we denote by $\Lambda =\mathrm{diag}\left( \beta _{1},...,\beta
_{r}\right) $, we get 
\begin{equation*}
E\left[ {\sum\limits_{i=1}^{k}}\beta _{i}Z_{i}^{2}\right] ={%
\sum\limits_{i=1}^{k}}\beta _{i}=\mathrm{trace}\left( \Lambda \right) =%
\mathrm{trace}\left( \boldsymbol{A}\left( \boldsymbol{\theta }\right) 
\boldsymbol{G}_{\ast }(\boldsymbol{\theta })\right) .
\end{equation*}%
The test given by the statistic $^{2}T_{\phi ,n}(f_{\widehat{\boldsymbol{%
\theta }}_{c}},f_{\widehat{\boldsymbol{\theta }}_{R}})$ is more conservative
than the one based on 
\begin{equation*}
^{3}T_{\phi ,n}(\widehat{\boldsymbol{\theta }}_{c}\mathbf{,}\widetilde{%
\boldsymbol{\theta }}_{rc})=\frac{^{2}T_{\phi ,n}(\widehat{\boldsymbol{%
\theta }}_{c}\mathbf{,}\widetilde{\boldsymbol{\theta }}_{rc})}{\nu }=\frac{%
T_{\phi ,n}(\widehat{\boldsymbol{\theta }}_{c}\mathbf{,}\widetilde{%
\boldsymbol{\theta }}_{rc})}{\nu \overline{\lambda }},
\end{equation*}%
and we can find $\nu $ by imposing the condition $Var\left[ ^{3}T_{\phi ,n}(%
\widehat{\boldsymbol{\theta }}_{c}\mathbf{,}\widetilde{\boldsymbol{\theta }}%
_{rc})\right] =2E\left[ ^{3}T_{\phi ,n}(\widehat{\boldsymbol{\theta }}_{c}%
\mathbf{,}\widetilde{\boldsymbol{\theta }}_{rc})\right] ,$ as in the
chi-squared distribution. Since%
\begin{equation*}
E\left[ ^{3}T_{\phi ,n}(\widehat{\boldsymbol{\theta }}_{c}\mathbf{,}%
\widetilde{\boldsymbol{\theta }}_{rc})\right] =\frac{r}{\nu }\text{ and }Var%
\left[ ^{3}T_{\phi ,n}(\widehat{\boldsymbol{\theta }}_{c}\mathbf{,}%
\widetilde{\boldsymbol{\theta }}_{rc})\right] =\frac{2r}{\nu },
\end{equation*}%
\begin{equation*}
\nu =1+{\sum\limits_{i=1}^{r}}\frac{\left( \beta _{i}-\overline{\lambda }%
\right) ^{2}}{r\overline{\lambda }^{2}}=1+CV^{2}(\{\beta _{i}\}_{i=1}^{r}),
\end{equation*}%
where $CV$ represents the coefficient of variation. Then a chi-square
distribution with $\frac{r}{\nu }$ degrees of freedom approximates the
asymptotic distribution of the statistic $^{3}T_{\phi ,n}(\widehat{%
\boldsymbol{\theta }}_{c}\mathbf{,}\widetilde{\boldsymbol{\theta }}_{rc})$
for large $n$.

The degrees of freedom of $^{3}T_{\phi ,n}(\widehat{\boldsymbol{\theta }}_{c}%
\mathbf{,}\widetilde{\boldsymbol{\theta }}_{rc})$ is $\frac{k}{\nu }$, which
may not be an integer. To avoid this difficulty one can modify the statistic
such that the first two moments match specifically with the $\chi _{r}^{2}$
distribution (rather than with just any other $\chi ^{2}$ distribution).
Specifically let 
\begin{equation}
X=\text{ }^{2}T_{\phi ,n}(\widehat{\boldsymbol{\theta }}_{c}\mathbf{,}%
\widetilde{\boldsymbol{\theta }}_{rc}).  \notag
\end{equation}%
We have 
\begin{align*}
E\left[ X\right] & =r=E\left[ \chi _{r}^{2}\right] , \\
Var[X]& =\frac{2{\sum\limits_{i=1}^{r}}\beta _{i}{}^{2}}{\overline{\lambda }%
^{2}}=2r+2{\sum\limits_{i=1}^{r}}\frac{\left( \beta _{i}-\overline{\lambda }%
\right) ^{2}}{\overline{\lambda }^{2}}=2r+c,
\end{align*}%
where $c$ stands for the last indicated term in the previous expression. We
define $Y=(X-a)/b$, where the constants $a$ and $b$ are such that 
\begin{equation*}
E(Y)=r,\quad Var(Y)=2r.
\end{equation*}%
Thus, 
\begin{equation*}
\frac{r-a}{b}=r,\quad \frac{2r+c}{b^{2}}=2r.
\end{equation*}%
Solving these equations, we get 
\begin{equation*}
b=\sqrt{1+\frac{c}{2r}},\quad a=r(1-b).
\end{equation*}%
Thus it makes sense to consider another modification of the statistic given
by 
\begin{equation*}
^{4}T_{\phi ,n}(\widehat{\boldsymbol{\theta }}_{c}\mathbf{,}\widetilde{%
\boldsymbol{\theta }}_{rc})=\frac{^{2}T_{\phi ,n}(\widehat{\boldsymbol{%
\theta }}_{c}\mathbf{,}\widetilde{\boldsymbol{\theta }}_{rc})-a}{b},
\end{equation*}%
the large sample distribution of which may be approximated by the $\chi
_{r}^{2}$ distribution.
\end{remark}

The approximation presented in this remark for the asymptotic distribution
of the $\phi $-divergence test-statistics, $T_{\phi ,n}(\widehat{\boldsymbol{%
\theta }}_{c}\mathbf{,}\widetilde{\boldsymbol{\theta }}_{rc})$, can be used
in the approximation of the $\phi $-divergence test-statistics $T_{\phi ,n}(%
\widehat{\boldsymbol{\theta }}_{c}\mathbf{,}\boldsymbol{\theta }_{0})$ as
well.

By the previous theorem, the null hypothesis should be rejected if $T_{\phi
,n}(\widehat{\boldsymbol{\theta }}_{c}\mathbf{,}\widetilde{\boldsymbol{%
\theta }}_{rc})\geq c_{\alpha }$, where $c_{\alpha }$ is the quantile of
order $1-\alpha $ of the asymptotic distribution of $T_{\phi ,n}(\widehat{%
\boldsymbol{\theta }}_{c}\mathbf{,}\widetilde{\boldsymbol{\theta }}_{rc})$.
The following theorem can be used to approximate the power function. Assume
that $\boldsymbol{\theta }\notin \Theta _{0}$ is the true value of the
parameter so that $\widehat{\boldsymbol{\theta }}_{c}\overset{a.s.}{\underset%
{n\rightarrow \infty }{\longrightarrow }}\boldsymbol{\theta }$ and that
there exists $\boldsymbol{\theta }^{\ast }\in \Theta _{0}$ such that the
restricted maximum composite likelihood estimator satisfies $\widetilde{%
\boldsymbol{\theta }}_{rc}\overset{a.s.}{\underset{n\rightarrow \infty }{%
\longrightarrow }}\boldsymbol{\theta }^{\ast }$, as well as 
\begin{equation*}
n^{1/2}\left( (\widehat{\boldsymbol{\theta }}_{c},\widetilde{\boldsymbol{%
\theta }}_{rc})-\left( \boldsymbol{\theta },\boldsymbol{\theta }^{\ast
}\right) \right) \overset{\mathcal{L}}{\underset{n\rightarrow \infty }{%
\longrightarrow }}\mathcal{N}\left( \left( 
\begin{array}{l}
\boldsymbol{0}_{p} \\ 
\boldsymbol{0}_{p}%
\end{array}%
\right) ,\left( 
\begin{array}{ll}
\boldsymbol{G}_{\ast }^{-1}(\boldsymbol{\theta }) & \boldsymbol{A}_{12}(%
\boldsymbol{\theta },\boldsymbol{\theta }^{\ast }) \\ 
\boldsymbol{A}_{12}^{T}(\boldsymbol{\theta },\boldsymbol{\theta }^{\ast }) & 
\boldsymbol{\Sigma }\left( \boldsymbol{\theta },\boldsymbol{\theta }^{\ast
}\right)%
\end{array}%
\right) \right) ,
\end{equation*}%
where $\boldsymbol{A}_{12}(\boldsymbol{\theta },\boldsymbol{\theta }^{\ast
}) $ and $\boldsymbol{\Sigma }\left( \boldsymbol{\theta },\boldsymbol{\theta 
}^{\ast }\right) $ are appropriate $p\times p$ matrices . We have then the
following result.

\begin{theorem}
Under $H_{1}$ we have 
\begin{equation*}
n^{1/2}\left( D_{\phi }(\widehat{\boldsymbol{\theta }}_{c},\widetilde{%
\boldsymbol{\theta }}_{rc})-D_{\phi }(\boldsymbol{\theta },\boldsymbol{%
\theta }^{\ast })\right) \overset{\mathcal{L}}{\underset{n\rightarrow \infty 
}{\longrightarrow }}\mathcal{N}\left( 0,\sigma ^{2}\left( \mathbf{%
\boldsymbol{\theta }},\boldsymbol{\theta }^{\ast }\right) \right)
\end{equation*}%
where 
\begin{equation}
\sigma ^{2}\left( \mathbf{\boldsymbol{\theta }},\boldsymbol{\theta }^{\ast
}\right) =\boldsymbol{t}^{T}\boldsymbol{G}_{\ast }^{-1}(\boldsymbol{\theta })%
\boldsymbol{t+}2\boldsymbol{t}^{T}\boldsymbol{A}_{12}(\boldsymbol{\theta },%
\boldsymbol{\theta }^{\ast })\boldsymbol{s+s}^{T}\boldsymbol{\Sigma }\left( 
\boldsymbol{\theta },\boldsymbol{\theta }^{\ast }\right) \boldsymbol{s}
\label{P1}
\end{equation}%
being%
\begin{equation*}
\boldsymbol{t=}\left. \frac{\partial D_{\phi }(\boldsymbol{\theta }_{1},%
\boldsymbol{\theta }^{\ast })}{\partial \boldsymbol{\theta }_{1}}\right\vert
_{\boldsymbol{\theta }_{1}=\mathbf{\boldsymbol{\theta }}}\text{ and }%
\boldsymbol{s=}\left. \frac{\partial D_{\phi }(\boldsymbol{\theta }_{2},%
\boldsymbol{\theta }^{\ast })}{\partial \boldsymbol{\theta }_{2}}\right\vert
_{\boldsymbol{\theta }_{2}=\boldsymbol{\theta }^{\ast }}.
\end{equation*}
\end{theorem}

\begin{proof}
The result follows in a straightforward manner by considering a first order
Taylor expansion of $D_{\phi }(\widehat{\boldsymbol{\theta }}_{c},\widetilde{%
\boldsymbol{\theta }}_{rc})$, which yields 
\begin{equation*}
D_{\phi }(\widehat{\boldsymbol{\theta }}_{c},\widetilde{\boldsymbol{\theta }}%
_{rc})=D_{\phi }(\boldsymbol{\theta },\boldsymbol{\theta }^{\ast })+%
\boldsymbol{t}^{T}(\widehat{\boldsymbol{\theta }}_{c}-\boldsymbol{\theta })+%
\boldsymbol{s}^{T}(\widetilde{\boldsymbol{\theta }}_{rc}-\boldsymbol{\theta }%
^{\ast })+o(\left\Vert \widehat{\boldsymbol{\theta }}_{c}-\boldsymbol{\theta 
}^{\ast }\right\Vert +\left\Vert \widetilde{\boldsymbol{\theta }}_{rc}-%
\boldsymbol{\theta }_{0}\right\Vert ).
\end{equation*}
\end{proof}

\begin{remark}
On the basis of the previous theorem we can get an approximation of the
power function 
\begin{equation*}
\pi _{n}^{\phi }\left( \mathbf{\boldsymbol{\theta }}\right) =\text{Pr}_{%
\mathbf{\boldsymbol{\theta }}}\left( T_{\phi ,n}(\widehat{\boldsymbol{\theta 
}}_{c},\widetilde{\boldsymbol{\theta }}_{rc})\geq c\right)
\end{equation*}%
at $\boldsymbol{\theta }$ as 
\begin{equation}
\pi _{n,\alpha }^{\beta ,\gamma }\left( \mathbf{\boldsymbol{\theta }}\right)
=1-\Phi \left( \frac{n^{1/2}}{\sigma \left( \mathbf{\boldsymbol{\theta }},%
\boldsymbol{\theta }^{\ast }\right) }\left( \frac{c}{2n}-D_{\phi }(%
\boldsymbol{\theta },\boldsymbol{\theta }^{\ast })\right) \right) ,
\label{P2}
\end{equation}%
where $\Phi (x)$ is the standard normal distribution function and $\sigma
^{2}\left( \mathbf{\boldsymbol{\theta }},\boldsymbol{\theta }^{\ast }\right) 
$ was defined in (\ref{P1}).

If some $\mathbf{\boldsymbol{\theta }}$ $\neq \boldsymbol{\theta }^{\ast }$
is the true parameter, then the probability of rejecting $H_{0}$ with the
rule that it is rejected when $T_{\phi ,n}(\widehat{\boldsymbol{\theta }}%
_{c},\widetilde{\boldsymbol{\theta }}_{rc})\geq c$ for a fixed test size $%
\alpha $ tends to one as $n\rightarrow \infty $. The test-statistic is
consistent in the Fraser's sense.

Obtaining the approximate sample size $n$ to guarantee a power of $\pi $ at
a given alternative $\mathbf{\boldsymbol{\theta }}^{\ast }$ is an
interesting application of formula (\ref{P2}). Let $n^{\ast }$ be the
positive root of the equation (\ref{P2}), i.e.%
\begin{equation*}
n^{\ast }=\frac{A+B+\sqrt{A(A+2B)}}{2D_{\phi }^{2}(\boldsymbol{\theta },%
\boldsymbol{\theta }^{\ast })},
\end{equation*}%
where 
\begin{equation*}
A=\sigma ^{2}\left( \mathbf{\boldsymbol{\theta }},\boldsymbol{\theta }^{\ast
}\right) \left( \Phi ^{-1}\left( 1-\pi \right) \right) ^{2}\quad \text{and}%
\quad B=cD_{\phi }(f_{\boldsymbol{\theta }},f_{\boldsymbol{\theta }^{\ast
}}).
\end{equation*}%
Then the required sample size is $n=\left[ n^{\ast }\right] +1$, where $%
\left[ \cdot \right] $ is used to denote \textquotedblleft integer part
of\textquotedblright .
\end{remark}

\begin{remark}
The class of $\phi $-divergence measures is a wide family of divergence
measures but unfortunately there are some classical divergence measures that
are not included in this family of $\phi $-divergence measures such as the R%
\'{e}nyi's divergence or the Sharma and Mittal's divergence. The expression
of R\'{e}nyi's divergence is given by 
\begin{equation*}
D_{\text{\textrm{R\'{e}nyi}}}^{a}(\widehat{\boldsymbol{\theta }}_{c}\mathbf{,%
}\widetilde{\boldsymbol{\theta }}_{rc})=\frac{1}{a\left( a-1\right) }\log
\dint_{R^{m}}\mathcal{CL}^{a}(\widehat{\boldsymbol{\theta }}_{c}\mathbf{,}%
\boldsymbol{y}\mathbf{)}\mathcal{CL}^{1-a}(\widetilde{\boldsymbol{\theta }}%
_{rc}\mathbf{,}\boldsymbol{y}\mathbf{)}d\boldsymbol{y}\text{ , if }a\neq 0,1,
\end{equation*}%
with 
\begin{equation*}
D_{\text{\textrm{R\'{e}nyi}}}^{0}(\widehat{\boldsymbol{\theta }}_{c},%
\widetilde{\boldsymbol{\theta }}_{rc})=\lim_{a\rightarrow 0}D_{\text{\textrm{%
R\'{e}nyi}}}^{a}(\widehat{\boldsymbol{\theta }}_{c}\mathbf{,}\widetilde{%
\boldsymbol{\theta }}_{rc})=D_{\mathrm{Kull}}(\widetilde{\boldsymbol{\theta }%
}_{rc},\widehat{\boldsymbol{\theta }}_{c})=\dint_{R^{m}}\mathcal{CL}(%
\widehat{\boldsymbol{\theta }}_{c}\mathbf{,}\boldsymbol{y}\mathbf{)}\log 
\frac{\mathcal{CL}(\widehat{\boldsymbol{\theta }}_{c}\mathbf{,}\boldsymbol{y}%
\mathbf{)}}{\mathcal{CL}(\widetilde{\boldsymbol{\theta }}_{rc}\mathbf{,}%
\boldsymbol{y}\mathbf{)}}d\boldsymbol{y}
\end{equation*}%
and 
\begin{equation*}
D_{\text{\textrm{R\'{e}nyi}}}^{1}(\widehat{\boldsymbol{\theta }}_{c},%
\widetilde{\boldsymbol{\theta }}_{rc})=\lim_{a\rightarrow 1}D_{\text{\textrm{%
R\'{e}nyi}}}(\widehat{\boldsymbol{\theta }}_{c}\mathbf{,}\widetilde{%
\boldsymbol{\theta }}_{rc})=D_{\mathrm{Kull}}(\widehat{\boldsymbol{\theta }}%
_{c}\mathbf{,}\widetilde{\boldsymbol{\theta }}_{rc}).
\end{equation*}%
This measure of divergence was introduced in R\'{e}nyi (1961) for $a>0$ and $%
a\neq 1$ and Liese and Vajda (1987) extended it for all $a\neq 1,0$. An
interesting divergence measure related to R\'{e}nyi divergence measure is
the Bhattacharya divergence defined as the R\'{e}nyi divergence for $a=1/2$
divided by $4$. Other interesting example of divergence measure, not
included in the family of $\phi $-divergence measures, is the divergence
measures introduced by Sharma and Mittal (1997).
\end{remark}

In order to unify the previous divergence measures as well as another
divergence measures Men\'{e}ndez et al. (1995, 1997) introduced the family
of divergences called \textquotedblleft $(h,\phi )$-divergence
measures\textquotedblright\ in the following way%
\begin{equation*}
D_{\phi }^{h}(\widehat{\boldsymbol{\theta }}_{c},\widetilde{\boldsymbol{%
\theta }}_{rc})=h\left( D_{\phi }(\widehat{\boldsymbol{\theta }}_{c},%
\widetilde{\boldsymbol{\theta }}_{rc})\right) ,
\end{equation*}%
where $h$ is a differentiable increasing function mapping from $\left[
0,\phi \left( 0\right) +\lim_{t\rightarrow \infty }\frac{\phi \left(
t\right) }{t}\right] $ onto $\left[ 0,\infty \right) $, with $h(0)=0$, $%
h^{\prime }(0)>0$, and $\phi \in \Psi $. In the next table these divergence
measures are presented, along with the corresponding expressions of $h$ and $%
\phi $.\newline
\begin{table}[htbp] \tabcolsep0.8pt  \centering%
$%
\begin{tabular}{ccccc}
\hline
Divergence & \hspace*{0.5cm} & $h\left( x\right) $ & \hspace*{0.5cm} & $\phi
\left( x\right) $ \\ \hline
\multicolumn{1}{l}{R\'{e}nyi} &  & \multicolumn{1}{l}{$\frac{1}{a\left(
a-1\right) }\log \left( a\left( a-1\right) x+1\right) ,\quad a\neq 0,1$} & 
& \multicolumn{1}{l}{$\frac{x^{a}-a\left( x-1\right) -1}{a\left( a-1\right) }%
,\quad a\neq 0,1$} \\ 
\multicolumn{1}{l}{Sharma-Mittal} &  & \multicolumn{1}{l}{$\frac{1}{b-1}%
\left\{ [1+a\left( a-1\right) x]^{\frac{b-1}{a-1}}-1\right\} ,\quad b,a\neq
1 $} &  & \multicolumn{1}{l}{$\frac{x^{a}-a\left( x-1\right) -1}{a\left(
a-1\right) },\quad a\neq 0,1$} \\ \hline
\end{tabular}%
\ $%
\end{table}%
\newline

Based on the $(h,\phi )$-divergence measures we can define a new family of $%
(h,\phi )$-divergence test-statistics for testing the null hypothesis $H_{0}$
given in (\ref{H}) 
\begin{equation}
T_{h,\phi ,n}(\widehat{\boldsymbol{\theta }}_{c},\widetilde{\boldsymbol{%
\theta }}_{rc})=\frac{2n}{\phi ^{\prime \prime }(1)h^{\prime }(0)}h\left(
D_{\phi }(\widehat{\boldsymbol{\theta }}_{c}\mathbf{,}\widetilde{\boldsymbol{%
\theta }}_{rc})\right) .  \label{Ah}
\end{equation}%
Since%
\begin{equation*}
h(x)=h(0)+h^{\prime }(0)x+o(x)
\end{equation*}%
the asymptotic distribution of $T_{h,\phi ,n}(\widehat{\boldsymbol{\theta }}%
_{c},\widetilde{\boldsymbol{\theta }}_{rc})$ coincides with the asymptotic
distribution of $T_{\phi ,n}(\widehat{\boldsymbol{\theta }}_{c},\widetilde{%
\boldsymbol{\theta }}_{rc})$. In a similar way we can define the family of $%
(h,\phi )$-divergence test-statistics for testing the null hypothesis $H_{0}$
given in (\ref{H1}) by%
\begin{equation*}
T_{h,\phi ,n}(\widehat{\boldsymbol{\theta }}_{c}\mathbf{,}\boldsymbol{\theta 
}_{0})=\frac{2n}{\phi ^{\prime \prime }(1)h^{\prime }(0)}h\left( D_{\phi }(%
\widehat{\boldsymbol{\theta }}_{c}\mathbf{,}\boldsymbol{\theta }_{0})\right)
.
\end{equation*}

\section{Numerical Example\label{Sec6}}

In this section we shall consider an example, studied previously by Xu and
Reid (2011) on the robustness of maximum composite estimator. The aim of
this section is to clarify the different issues which are discussed in the
previous sections.

Consider the random vector $\boldsymbol{Y}=(Y_{1},Y_{2},Y_{3},Y_{4})^{T}$
which follows a four dimensional normal distribution with mean vector $%
\boldsymbol{\mu }=(\mu _{1},\mu _{2},\mu _{3},\mu _{4})^{T}$ and
variance-covariance matrix%
\begin{equation}
\Sigma =\left( 
\begin{array}{cccc}
1 & \rho & 2\rho & 2\rho \\ 
\rho & 1 & 2\rho & 2\rho \\ 
2\rho & 2\rho & 1 & \rho \\ 
2\rho & 2\rho & \rho & 1%
\end{array}%
\right) ,  \label{e4.1}
\end{equation}%
i.e., we suppose that the correlation between $Y_{1}$ and $Y_{2}$ is the
same as the correlation between $Y_{3}$ and $Y_{4}$. Taking into account
that $\Sigma $\ must be semi-definite positive, the following condition is
imposed, $-\frac{1}{5}\leq \rho \leq \frac{1}{3}$. In order to avoid several
problems regarding the consistency of the maximum likelihood estimator of
the parameter $\rho $ (cf. Xu and Reid, 2011), we shall consider the
composite likelihood function%
\begin{equation*}
\mathcal{CL}(\boldsymbol{\theta },\boldsymbol{y})=f_{A_{1}}(\boldsymbol{%
\theta },\boldsymbol{y})f_{A_{2}}(\boldsymbol{\theta },\boldsymbol{y}),
\end{equation*}%
where%
\begin{align*}
f_{A_{1}}(\boldsymbol{\theta },\boldsymbol{y})& =f_{12}(\mu _{1},\mu
_{2},\rho ,y_{1},y_{2}), \\
f_{A_{2}}(\boldsymbol{\theta },\boldsymbol{y})& =f_{34}(\mu _{3},\mu
_{4},\rho ,y_{3},y_{4}),
\end{align*}%
where $f_{12}$ and $f_{34}$ are the densities of the marginals of $%
\boldsymbol{Y}$, i.e. bivariate normal distributions with mean vectors $(\mu
_{1},\mu _{2})^{T}$ and $(\mu _{3},\mu _{4})^{T}$, respectively, and common
variance-covariance matrix%
\begin{equation*}
\left( 
\begin{array}{cc}
1 & \rho \\ 
\rho & 1%
\end{array}%
\right) ,
\end{equation*}%
with expressions given by%
\begin{equation*}
f_{h,h+1}(\mu _{h},\mu _{h+1},\rho ,y_{h},y_{h+1})=\frac{1}{2\pi \sqrt{%
1-\rho ^{2}}}\exp \left\{ -\tfrac{1}{2(1-\rho ^{2})}Q(y_{h},y_{h+1})\right\}
,\text{ }h\in \{1,3\},
\end{equation*}%
being%
\begin{equation*}
Q(y_{h},y_{h+1})=(y_{h}-\mu _{h})^{2}-2\rho (y_{h}-\mu _{h})(y_{h+1}-\mu
_{h+1})+(y_{h+1}-\mu _{h+1})^{2},\text{ }h\in \{1,3\}.
\end{equation*}

In this context, the interest is focused in testing the composite null
hypothesis hypotheses%
\begin{equation}
H_{0}:\rho =\rho _{0}\text{\quad against\quad }H_{1}:\rho \neq \rho _{0},
\label{test}
\end{equation}%
by using the composite $\phi $-divergence test-statistics, presented above.
In this case, the parameter space is given by%
\begin{equation*}
\Theta =\left\{ \boldsymbol{\theta }=(\mu _{1},\mu _{2},\mu _{3},\mu
_{4},\rho )^{T}:\mu _{i}\in 
\mathbb{R}
,i=1,...,4\text{ and }-\tfrac{1}{5}\leq \rho \leq \tfrac{1}{3}\right\} .
\end{equation*}%
If we consider $g:\Theta \subseteq 
\mathbb{R}
^{5}\longrightarrow 
\mathbb{R}
$, with%
\begin{equation}
g(\boldsymbol{\theta })=g((\mu _{1},\mu _{2},\mu _{3},\mu _{4},\rho
)^{T})=\rho -\rho _{0},  \label{e4.2}
\end{equation}%
the parameter space under the null hypothesis is given by%
\begin{equation*}
\Theta _{0}=\left\{ \boldsymbol{\theta }=(\mu _{1},\mu _{2},\mu _{3},\mu
_{4},\rho )^{T}\in \Theta :g(\mu _{1},\mu _{2},\mu _{3},\mu _{4},\rho
)=0\right\} .
\end{equation*}%
It is now clear that the dimensions of both parameter spaces are $\dim
(\Theta )=5$ and $\dim (\Theta _{0})=4$. Consider now a random sample of
size $n$, $\boldsymbol{y}_{i}=(y_{i1},...,y_{i4})^{T}$, $i=1,...,n$. The
maximum composite likelihood estimators of the parameters $\mu _{i}$, $%
i=1,2,3,4$, and $\rho $ in $\Theta $ are obtained by standard maximization
of the composite log-density function associated to the random sample of
size $n$,%
\begin{align*}
c\ell (\boldsymbol{\theta },\boldsymbol{y}_{1},...,\boldsymbol{y}_{n})&
=\sum_{i=1}^{n}c\ell (\boldsymbol{\theta },\boldsymbol{y}_{i})=%
\sum_{i=1}^{n}\log \mathcal{CL}(\boldsymbol{\theta },\boldsymbol{y}_{i}) \\
& =\sum_{i=1}^{n}\left[ \log f_{A_{1}}(\boldsymbol{\theta },\boldsymbol{y}%
_{i})+\log f_{A_{2}}(\boldsymbol{\theta },\boldsymbol{y}_{i})\right] \\
& =\sum_{i=1}^{n}\ell _{A_{1}}(\boldsymbol{\theta },\boldsymbol{y}%
_{i})+\sum_{i=1}^{n}\ell _{A_{2}}(\boldsymbol{\theta },\boldsymbol{y}_{i}) \\
& =-\frac{n}{2}\log \left( 1-\rho ^{2}\right) -\frac{1}{2\left( 1-\rho
^{2}\right) }(\varsigma _{1}^{2}+\varsigma _{2}^{2}+\varsigma
_{3}^{2}+\varsigma _{4}^{2}-2\rho (\varsigma _{12}+\varsigma _{34}))+k,
\end{align*}%
where%
\begin{align*}
\varsigma _{j}^{2}& =\sum_{i=1}^{n}(y_{ij}-\mu _{j})^{2},\text{ }j\in
\{1,2,3,4\}, \\
\varsigma _{h,h+1}& =\sum_{i=1}^{n}(y_{ih}-\mu _{h})(y_{i,h+1}-\mu _{h+1}),%
\text{ }h\in \{1,3\},
\end{align*}%
and $k$ is a constant no dependent of the unknown parameters, i.e., by
solving the system of the following system of five equations%
\begin{align*}
(\overline{y}_{1}-\mu _{1})-\rho (\overline{y}_{2}-\mu _{2})& =0, \\
(\overline{y}_{2}-\mu _{2})-\rho (\overline{y}_{1}-\mu _{1})& =0, \\
(\overline{y}_{3}-\mu _{3})-\rho (\overline{y}_{4}-\mu _{4})& =0, \\
(\overline{y}_{4}-\mu _{4})-\rho (\overline{y}_{3}-\mu _{3})& =0, \\
n\rho ^{3}-\frac{\varsigma _{12}+\varsigma _{34}}{2}\rho ^{2}+\left( \frac{%
\varsigma _{1}^{2}+\varsigma _{2}^{2}+\varsigma _{3}^{2}+\varsigma _{4}^{2}}{%
2}-n\right) \rho -\frac{\varsigma _{12}+\varsigma _{34}}{2}& =0,
\end{align*}%
with%
\begin{equation}
\overline{y}_{j}=\frac{1}{n}\dsum\limits_{i=1}^{n}y_{ij},\text{ }j\in
\{1,2,3,4\}.  \label{sm}
\end{equation}%
From the first two equations we get%
\begin{align*}
\rho (\overline{y}_{1}-\mu _{1})-\rho ^{2}(\overline{y}_{2}-\mu _{2})& =0, \\
-\rho (\overline{y}_{1}-\mu _{1})+\overline{y}_{2}-\mu _{2}& =0.
\end{align*}%
Therefore%
\begin{equation*}
\left( 1-\rho ^{2}\right) (\overline{y}_{2}-\mu _{2})=0,
\end{equation*}%
and since we assume that $\rho \in (-\tfrac{1}{5},\tfrac{1}{3})\subset
(-1,1) $, it is obtained that%
\begin{equation*}
\widehat{\mu }_{1}=\overline{y}_{1}\text{\quad and\quad }\widehat{\mu }_{2}=%
\overline{y}_{2}.
\end{equation*}%
In a similar manner, from the third and fourth equations we can get that%
\begin{equation*}
\widehat{\mu }_{3}=\overline{y}_{3}\text{\quad and\quad }\widehat{\mu }_{4}=%
\overline{y}_{4}.
\end{equation*}%
The maximum composite likelihood estimator of $\rho $ under $\Theta $, $%
\widehat{\rho }$, is the real solution of the following cubic equation%
\begin{equation*}
\rho ^{3}-\frac{v_{12}+v_{34}}{2}\rho ^{2}+\left( \frac{%
v_{1}^{2}+v_{2}^{2}+v_{3}^{2}+v_{4}^{2}}{2}-1\right) \rho -\frac{%
v_{12}+v_{34}}{2}=0,
\end{equation*}%
where%
\begin{align*}
v_{j}^{2}& =\frac{1}{n}\widehat{\varsigma }_{j}^{2}=\frac{1}{n}%
\sum_{i=1}^{n}(y_{ij}-\overline{y}_{j})^{2},\text{ }j\in \{1,2,3,4\}, \\
v_{h,h+1}& =\frac{1}{n}\widehat{\varsigma }_{h,h+1}=\frac{1}{n}%
\sum_{i=1}^{n}(y_{ih}-\overline{y}_{h})(y_{i,h+1}-\overline{y}_{h+1}),\text{ 
}h\in \{1,3\},
\end{align*}%
and $v_{j}^{2}$, $j\in \{1,2,3,4\}$, are sampling variances and $v_{h,h+1},$ 
$h\in \{1,3\}$, sampling covariances.

Under $\Theta _{0}$, the restricted maximum composite likelihood estimators
of the parameters $\mu _{j}$, $j\in \{1,2,3,4\}$ are given by,%
\begin{equation*}
\widetilde{\mu }_{j}=\overline{y}_{j},\text{ }j\in \{1,2,3,4\},
\end{equation*}%
with\ $\overline{y}_{i}$ given by (\ref{sm}). Therefore, in our model, the
maximum composite likelihood estimators are 
\begin{equation*}
\widehat{\boldsymbol{\theta }}_{c}=\left( \overline{y}_{1},\overline{y}_{2},%
\overline{y}_{3},\overline{y}_{4},\widehat{\rho }\right) ^{T}\text{ and }%
\widetilde{\boldsymbol{\theta }}_{rc}=\left( \overline{y}_{1},\overline{y}%
_{2},\overline{y}_{3},\overline{y}_{4},\rho _{0}\right) ^{T},
\end{equation*}%
under $\Theta $ and $\Theta _{0}$ respectively.

After some heavy algebraic manipulations the sensitivity or Hessian matrix $%
\boldsymbol{H}\mathbf{(}\boldsymbol{\theta }\mathbf{)}$ is given by%
\begin{equation}
\boldsymbol{H}\mathbf{(}\boldsymbol{\theta }\mathbf{)}=\frac{1}{1-\rho ^{2}}%
\left( 
\begin{array}{ccccc}
1 & -\rho & 0 & 0 & 0 \\ 
-\rho & 1 & 0 & 0 & 0 \\ 
0 & 0 & 1 & -\rho & 0 \\ 
0 & 0 & -\rho & 1 & 0 \\ 
0 & 0 & 0 & 0 & 2\frac{1+\rho ^{2}}{1-\rho ^{2}}%
\end{array}%
\right) .  \label{e4.3}
\end{equation}
In a similar manner, the expression of the variability matrix $\boldsymbol{J}%
\mathbf{(}\boldsymbol{\theta }\mathbf{)}$ coincides with that of sensitivity
matrix $\boldsymbol{H}\mathbf{(}\boldsymbol{\theta }\mathbf{)}$, i.e. $%
\boldsymbol{J}\mathbf{(}\boldsymbol{\theta }\mathbf{)}=\boldsymbol{H}\mathbf{%
(}\boldsymbol{\theta }\mathbf{)}$.

In order to get the unique non zero eigenvalue $\beta _{1}$ from Theorem \ref%
{Theorem4}, it is necessary to obtain (\ref{G}), which, in the present
setup, is given by%
\begin{equation}
\boldsymbol{G}\mathbf{(}\boldsymbol{\theta }\mathbf{)}=\frac{\partial g(%
\boldsymbol{\theta })}{\partial \boldsymbol{\theta }}=(0,0,0,0,1)^{T},
\label{e4.5}
\end{equation}%
where $g(\boldsymbol{\theta })$ is given by (\ref{e4.2}) in the context of
the present example. In addition, taking into account that $\boldsymbol{Q}%
\mathbf{(}\boldsymbol{\theta }\mathbf{)}=-\boldsymbol{G}\mathbf{(}%
\boldsymbol{\theta }\mathbf{)}$ and after some algebra, it is concluded that 
$\beta _{1}=1$\ and therefore the asymptotic distribution of the composite $%
\phi $-divergence test-statistics is%
\begin{equation*}
T_{\phi ,n}(\widehat{\boldsymbol{\theta }}_{c}\mathbf{,}\widetilde{%
\boldsymbol{\theta }}_{rc})=\frac{2n}{\phi ^{\prime \prime }(1)}D_{\phi }(%
\widehat{\boldsymbol{\theta }}_{c},\widetilde{\boldsymbol{\theta }}_{rc})%
\overset{\mathcal{L}}{\underset{n\rightarrow \infty }{\longrightarrow }}%
{\large \chi }_{1}^{2},
\end{equation*}%
under the null hypothesis $H_{0}:$ $\rho =\rho _{0}$. In a completely
similar manner, the composite likelihood ratio test, presented in Remark 1,
for testing $H_{0}:$ $\rho =\rho _{0}$, is 
\begin{equation*}
\lambda _{n}(\widehat{\boldsymbol{\theta }}_{c},\widetilde{\boldsymbol{%
\theta }}_{rc})=2\left( c\ell (\widehat{\boldsymbol{\theta }}_{c})-c\ell (%
\widetilde{\boldsymbol{\theta }}_{rc})\right) \overset{\mathcal{L}}{\underset%
{n\rightarrow \infty }{\longrightarrow }}{\large \chi }_{1}^{2},
\end{equation*}%
because the only non zero eigenvalue of the asymptotic distribution (\ref{R1}%
) is equal to one.

In a similar way, if we consider the composite $(h,\phi )$-divergence
test-statistics, we have 
\begin{equation*}
T_{h,\phi ,n}(\widehat{\boldsymbol{\theta }}_{c},\widetilde{\boldsymbol{%
\theta }}_{rc})=\frac{2n}{\phi ^{\prime \prime }(1)h^{\prime }(0)}h\left(
D_{\phi }(\widehat{\boldsymbol{\theta }}_{c}\mathbf{,}\widetilde{\boldsymbol{%
\theta }}_{rc})\right) \overset{\mathcal{L}}{\underset{n\rightarrow \infty }{%
\longrightarrow }}\chi _{1}^{2}.
\end{equation*}%
I order to apply the above theoretic issues in practice it is necessary to
consider a particular convex function $\phi $ in order to get a concrete $%
\phi $-divergence or to consider $\phi $ and $h$ in order to get an $(h,\phi
)$-divergence. Using the R\'{e}nyi's family of divergences, i.e, a family of 
$(h,\phi )$-divergences with $\phi $ and $h$ given in Table 1, the family of
test-statistics is given by%
\begin{align*}
T_{n}^{r}(\widehat{\boldsymbol{\theta }}_{c}\mathbf{,}\widetilde{\boldsymbol{%
\theta }}_{rc})& =\frac{2n}{r(r-1)}\left( \log \int_{%
\mathbb{R}
^{2}}\frac{f_{12}^{r}(\widehat{\mu }_{1},\widehat{\mu }_{2},\widehat{\rho }%
,y_{1},y_{2})}{f_{12}^{r-1}(\widehat{\mu }_{1},\widehat{\mu }_{2},\rho
_{0},y_{1},y_{2})}dy_{1}dy_{2}+\log \int_{%
\mathbb{R}
^{2}}\frac{f_{34}^{r}(\widehat{\mu }_{3},\widehat{\mu }_{4},\widehat{\rho }%
,y_{3},y_{4})}{f_{34}^{r-1}(\widehat{\mu }_{3},\widehat{\mu }_{4},\rho
_{0},y_{3},y_{4})}dy_{3}dy_{4}\right) \\
& =\frac{4n}{r(r-1)}\log \int_{%
\mathbb{R}
^{2}}\frac{f_{12}^{r}(\widehat{\mu }_{1},\widehat{\mu }_{2},\widehat{\rho }%
,y_{1},y_{2})}{f_{12}^{r-1}(\widehat{\mu }_{1},\widehat{\mu }_{2},\rho
_{0},y_{1},y_{2})}dy_{1}dy_{2},
\end{align*}%
for $r\neq 0,1$. The last equality follows becuase the integrals does not
depend on $\widehat{\mu }_{1},\widehat{\mu }_{2}$, $\widehat{\mu }_{3}$ and $%
\widehat{\mu }_{4}.$ For $r=1$ we have 
\begin{align*}
T_{n}^{1}(\widehat{\boldsymbol{\theta }}_{c}\mathbf{,}\widetilde{\boldsymbol{%
\theta }}_{rc})& =2n\left( \dint\nolimits_{%
\mathbb{R}
^{2}}f_{34}(\widehat{\mu }_{3},\widehat{\mu }_{4},\widehat{\rho }%
,y_{3},y_{4})dy_{3}dy_{4}\dint\nolimits_{%
\mathbb{R}
^{2}}f_{12}(\widehat{\mu }_{1},\widehat{\mu }_{2},\widehat{\rho }%
,y_{1},y_{2})\log \frac{f_{12}(\widehat{\mu }_{1},\widehat{\mu }_{2},%
\widehat{\rho },y_{1},y_{2})}{f_{12}(\widehat{\mu }_{1},\widehat{\mu }%
_{2},\rho _{0},y_{1},y_{2})}dy_{1}dy_{2}\right. \\
& \left. +\dint\nolimits_{%
\mathbb{R}
^{2}}f_{12}(\widehat{\mu }_{1},\widehat{\mu }_{2},\widehat{\rho }%
,y_{1},y_{2})dy_{1}dy_{2}\dint\nolimits_{%
\mathbb{R}
^{2}}f_{34}(\widehat{\mu }_{3},\widehat{\mu }_{4},\widehat{\rho }%
,y_{3},y_{4})\log \frac{f_{34}(\widehat{\mu }_{3},\widehat{\mu }_{4},%
\widehat{\rho },y_{3},y_{4})}{f_{34}(\widehat{\mu }_{3},\widehat{\mu }%
_{4},\rho _{0},y_{3},y_{4})}dy_{3}dy_{4}\right) \\
& =4n\dint\nolimits_{%
\mathbb{R}
^{2}}f_{12}(\widehat{\mu }_{1},\widehat{\mu }_{2},\widehat{\rho }%
,y_{1},y_{2})\log \frac{f_{12}(\widehat{\mu }_{1},\widehat{\mu }_{2},%
\widehat{\rho },y_{1},y_{2})}{f_{12}(\widehat{\mu }_{1},\widehat{\mu }%
_{2},\rho _{0},y_{1},y_{2})}dy_{1}dy_{2}
\end{align*}%
\begin{equation*}
T_{n}^{0}(\widehat{\boldsymbol{\theta }}_{c}\mathbf{,}\widetilde{\boldsymbol{%
\theta }}_{rc})=T_{n}^{1}(\widehat{\boldsymbol{\theta }}_{rc}\mathbf{,}%
\widetilde{\boldsymbol{\theta }}_{c}).
\end{equation*}%
Based on the expression for R\'{e}nyi divergence in normal populations (for
more details see Pardo, 2006, p. 33) we have%
\begin{equation}
T_{n}^{r}(\widehat{\boldsymbol{\theta }}_{c},\widetilde{\boldsymbol{\theta }}%
_{rc})=\left\{ 
\begin{array}{ll}
\frac{2n}{r(r-1)}\log \dfrac{(1-\rho _{0}^{2})^{r}(1-\widehat{\rho }%
^{2})^{-(r-1)}}{1-\left[ r\rho _{0}+(1-r)\widehat{\rho }\right] ^{2}}, & 
r\notin \{0,1\},\;\widehat{\rho }\in \left( \frac{r}{r-1}\rho _{0}-\frac{1}{%
\left\vert r-1\right\vert },\frac{r}{r-1}\rho _{0}+\frac{1}{\left\vert
r-1\right\vert }\right) , \\ 
+\infty , & r\notin \{0,1\},\;\widehat{\rho }\notin \left( \frac{r}{r-1}\rho
_{0}-\frac{1}{\left\vert r-1\right\vert },\frac{r}{r-1}\rho _{0}+\frac{1}{%
\left\vert r-1\right\vert }\right) , \\ 
2n\left( \log \frac{1-\rho _{0}^{2}}{1-\widehat{\rho }^{2}}+2\frac{\rho
_{0}(\rho _{0}-\widehat{\rho })}{1-\rho _{0}^{2}}\right) , & r=1, \\ 
2n\left( \log \frac{1-\widehat{\rho }^{2}}{1-\rho _{0}^{2}}+2\frac{\widehat{%
\rho }(\widehat{\rho }-\rho _{0})}{1-\widehat{\rho }^{2}}\right) , & r=0.%
\end{array}%
\right.  \label{Renyi}
\end{equation}%
Similarly, using the Cressie-Read's family of divergences, the family of
test-statistics is given by%
\begin{equation}
T_{n}^{\lambda }(\widehat{\boldsymbol{\theta }}_{c},\widetilde{\boldsymbol{%
\theta }}_{rc})=\left\{ 
\begin{array}{ll}
\frac{4n}{\lambda (\lambda +1)}\left( \sqrt{\dfrac{(1-\rho
_{0}^{2})^{\lambda +1}(1-\widehat{\rho }^{2})^{-\lambda }}{1-\left[ (\lambda
+1)\rho _{0}-\lambda \widehat{\rho }\right] ^{2}}}-1\right) , & \lambda
\notin \{0,-1\},\;\widehat{\rho }\in \left( \frac{\lambda +1}{\lambda }\rho
_{0}-\frac{1}{\left\vert \lambda \right\vert },\frac{\lambda +1}{\lambda }%
\rho _{0}+\frac{1}{\left\vert \lambda \right\vert }\right) , \\ 
+\infty , & \lambda \notin \{0,-1\},\;\widehat{\rho }\notin \left( \frac{%
\lambda +1}{\lambda }\rho _{0}-\frac{1}{\left\vert \lambda \right\vert },%
\frac{\lambda +1}{\lambda }\rho _{0}+\frac{1}{\left\vert \lambda \right\vert 
}\right) , \\ 
2n\left( \log \frac{1-\rho _{0}^{2}}{1-\widehat{\rho }^{2}}+2\frac{\rho
_{0}(\rho _{0}-\widehat{\rho })}{1-\rho _{0}^{2}}\right) =T_{n}^{1}(\widehat{%
\boldsymbol{\theta }}_{c},\widetilde{\boldsymbol{\theta }}_{rc}), & \lambda
=0, \\ 
2n\left( \log \frac{1-\widehat{\rho }^{2}}{1-\rho _{0}^{2}}+2\frac{\widehat{%
\rho }(\widehat{\rho }-\rho _{0})}{1-\widehat{\rho }^{2}}\right) =T_{n}^{0}(%
\widehat{\boldsymbol{\theta }}_{c},\widetilde{\boldsymbol{\theta }}_{rc}), & 
\lambda =-1.%
\end{array}%
\right.  \label{Cressie}
\end{equation}%
After some algebra we can also obtain the composite likelihood ratio test.
This has the following expression%
\begin{align}
\lambda _{n}(\widehat{\boldsymbol{\theta }}_{c},\widetilde{\boldsymbol{%
\theta }}_{rc})& =2\left( c\ell (\widehat{\boldsymbol{\theta }}_{c},%
\boldsymbol{y}_{1},...,\boldsymbol{y}_{n})-c\ell (\widetilde{\boldsymbol{%
\theta }}_{rc},\boldsymbol{y}_{1},...,\boldsymbol{y}_{n})\right)  \notag \\
& =2n\left[ \log \frac{1-\rho _{0}^{2}}{1-\widehat{\rho }^{2}}%
+(v_{1}^{2}+v_{2}^{2}+v_{3}^{2}+v_{4}^{2})\left( \frac{1}{1-\rho _{0}^{2}}-%
\frac{1}{1-\widehat{\rho }^{2}}\right) -2(v_{12}+v_{34})\left( \frac{\rho
_{0}}{1-\rho _{0}^{2}}-\frac{\widehat{\rho }}{1-\widehat{\rho }^{2}}\right) %
\right] .  \label{CLR}
\end{align}

\section{Simulation study\label{Sec7}}

In this section a simulation study is presented in order to study the
behavior of the composite $\phi $-divergence test-statistics. The
theoretical model studied in the previous section is followed by using the
composite Cressie-Read test-statistics (\ref{Cressie})$.$ The composite
likelihood ratio test-statistic (CLRT), given in (\ref{CLR}), is also
considered. A special attention has been paid to the hypothesis testing (\ref%
{test}) with $\rho _{0}\in \{-0.1,0.2\}$. The case $\rho _{0}=0$ has been
considered, but this case is less important since taking into account the
way of the theoretical model under consideration and having the case of
independent observations, the composite likelihood theory is useless. For
finite sample sizes and nominal size $\alpha =0.05$, the estimated
significance level for different composite Cressie-Read test-statistics as
well as for the CLRT, are given by 
\begin{equation*}
\alpha _{n}^{(\lambda )}(\rho _{0})=\Pr (T_{n}^{\lambda }(\widehat{%
\boldsymbol{\theta }}_{c},\widetilde{\boldsymbol{\theta }}_{rc})>\chi
_{1,0.05}^{2}|H_{0})\text{, }\lambda \in 
\mathbb{R}
\text{ and }\alpha _{n}^{(CLRT)}(\rho _{0})=\Pr \left( \lambda _{n}(\widehat{%
\boldsymbol{\theta }}_{c},\widetilde{\boldsymbol{\theta }}_{rc})>\chi
_{1,0.05}^{2}|H_{0}\right) .
\end{equation*}%
More thorougly, the composite Cressie-Read test-statistics with $\lambda $ $%
\in \left\{ -1,-0.5,0.2/3,1,1.5\right\} $\ have beed selected for the study.
Following Dale (1986), we consider the inequality 
\begin{equation}
\left\vert \text{logit}(1-\alpha _{n}^{(\bullet )})-\text{logit}(1-\alpha
)\right\vert \leq \varepsilon  \label{a}
\end{equation}%
where logit$(p)=$ $\ln (p/\left( 1-p\right) )$. By chosing $\varepsilon
=0.45 $ the composite test-statistics valid for the study are limited to
those verifying $\alpha _{n}^{(\lambda )}\in \left( 0.0325,0.07625\right) $.
This criterion has been used in many previous studies, see for instance
Cressie et al (2003), Mart\'{\i}n et al. (2014), Mart\'{\i}n and Pardo
(2012) and references therein.

Through $R=10,000$ replications of the simulation experiment, with the model
under the null hypothesis, the estimated significance level for different
composite Cressie-Read test-statistics are%
\begin{equation*}
\widehat{\alpha }_{n}^{(\lambda )}(\rho _{0})=\widehat{\Pr }(T_{n}^{\lambda
}(\widehat{\boldsymbol{\theta }}_{c},\widetilde{\boldsymbol{\theta }}%
_{rc})>\chi _{1,0.05}^{2}|H_{0})=\frac{\dsum\limits_{i=1}^{R}\emph{I}%
(T_{n,i}^{\lambda }(\widehat{\boldsymbol{\theta }}_{c},\widetilde{%
\boldsymbol{\theta }}_{rc})>\chi _{1,0.05}^{2}|\rho _{0})}{R},
\end{equation*}%
with $\emph{I}(S)$ being and indicator function (with value $1$ if $S$ is
true and $0$ otherwise) and the estimated significance level for CLRT%
\begin{equation*}
\widehat{\alpha }_{n}^{(CLRT)}(\rho _{0})=\widehat{\Pr }(\lambda _{n}(%
\widehat{\boldsymbol{\theta }}_{c},\widetilde{\boldsymbol{\theta }}%
_{rc})>\chi _{1,0.05}^{2}|H_{0})=\frac{\dsum\limits_{i=1}^{R}\emph{I}%
(\lambda _{n}(\widehat{\boldsymbol{\theta }}_{c},\widetilde{\boldsymbol{%
\theta }}_{rc})>\chi _{1,0.05}^{2}|\rho _{0})}{R}.
\end{equation*}

In Table \ref{table1} we present the simulated level for different values of 
$\lambda $ $\in \left\{ -1,-0.5,0.2/3,1,1.5\right\} $ as well as for the
CLRT, when $n=100,$ $n=200$ and $n=300$ for $\rho _{0}=-0.1$ and $\rho
_{0}=0.2$. In order to investigate the behavior for $\rho _{0}=0$ we present
in Table \ref{table2} the simulated level for $\lambda $ $\in \left\{
-1,-0.5,0.2/3,1,1.5\right\} $ as well as the simulated level of CLRT for $%
n=50$, $n=100$, $n=200$ and $n=300$. Clearly, as expected the performance of
the traditional divergence and likelihood methods is stronger in comparison
with the composite divergence and likelihood methods.

\begin{table}[hbpt]  \tabcolsep2.8pt \small\centering%
$%
\begin{tabular}{ccccccc}
\hline
& \multicolumn{2}{c}{$n=100$} & \multicolumn{2}{c}{$n=200$} & 
\multicolumn{2}{c}{$n=300$} \\ \hline
& $\rho _{0}=-0.1$ & $\rho _{0}=0.2$ & $\rho _{0}=-0.1$ & $\rho _{0}=0.2$ & $%
\rho _{0}=-0.1$ & $\rho _{0}=0.2$ \\ \hline
$CLRT$ & $0.0688$ & $0.0694$ & $0.0673$ & $0.0687$ & $0.0645$ & $0.0662$ \\ 
$\lambda =-1$ & $0.0756$ & $0.0762$ & $0.0706$ & $0.0740$ & $0.0666$ & $%
0.0685$ \\ 
$\lambda =-0.5$ & $0.0738$ & $0.0746$ & $0.0697$ & $0.0727$ & $0.0662$ & $%
0.0670$ \\ 
$\lambda =0$ & $0.0725$ & $0.0739$ & $0.0691$ & $0.0720$ & $0.0659$ & $%
0.0672 $ \\ 
$\lambda =2/3$ & $0.0726$ & $0.0739$ & $0.0694$ & $0.0719$ & $0.0662$ & $%
0.0677$ \\ 
$\lambda =1$ & $0.0739$ & $0.0747$ & $0.0700$ & $0.0720$ & $0.0662$ & $%
0.0680 $ \\ 
$\lambda =1.5$ & $0.0762$ & $0.0769$ & $0.0711$ & $0.0729$ & $0.0674$ & $%
0.0677$ \\ \hline
\end{tabular}%
$%
\caption{Simulated significance level for $\rho_0=-0.1$ and
$\rho_0=0.2$.\label{table1}}%
\end{table}%

\begin{table}[hbpt]  \tabcolsep2.8pt \small\centering%
$%
\begin{tabular}{ccccc}
\hline
& $n=50$ & $n=100$ & $n=200$ & $n=300$ \\ \hline
$LRT$ & $0.0543$ & $0.0529$ & $0.0527$ & $0.0526$ \\ 
$\lambda =-1$ & $0.0707$ & $0.0605$ & $0.0559$ & $0.0542$ \\ 
$\lambda =-0.5$ & $0.0677$ & $0.0594$ & $0.0553$ & $0.0540$ \\ 
$\lambda =0$ & $0.0659$ & $0.0577$ & $0.0552$ & $0.0540$ \\ 
$\lambda =2/3$ & $0.0670$ & $0.0591$ & $0.0552$ & $0.0540$ \\ 
$\lambda =1$ & $0.0686$ & $0.0597$ & $0.0553$ & $0.0541$ \\ 
$\lambda =1.5$ & $0.0726$ & $0.0610$ & $0.0564$ & $0.0544$ \\ \hline
\end{tabular}%
$\caption{Simulated significance level for $\rho_0=0$.\label{table2}}%
\end{table}%

For finite sample sizes and nominal size $\alpha =0.05$, the simulated
powers are obtained under $H_{1}$ in (\ref{test}), when $\rho \in
\{-0.2,-0.15,0,0.1\}$ and $\rho _{0}=-0.1$ (Table \ref{table3}) and when $%
\rho \in \{0,0.15,0.25,0.3\}$ and $\rho _{0}=0.2$ (Table \ref{table4}). The
(simulated) power for different composite Cressie-Read test-statistics is
obtained by 
\begin{equation*}
\beta _{n}^{(\lambda )}(\rho _{0},\rho )=\Pr (T_{n}^{\lambda }(\widehat{%
\boldsymbol{\theta }}_{c},\widetilde{\boldsymbol{\theta }}_{rc})>\chi
_{1,0.05}^{2}|H_{1})\text{ and }\widehat{\beta }_{n}^{(\lambda )}(\rho
_{0},\rho )=\frac{\dsum\limits_{i=1}^{R}\emph{I}(T_{n}^{\lambda }(\widehat{%
\boldsymbol{\theta }}_{c},\widetilde{\boldsymbol{\theta }}_{rc})>\chi
_{1,0.05}^{2}|\rho _{0},\rho )}{R},
\end{equation*}%
and for the CLRT by%
\begin{equation*}
\beta _{n}^{(CLRT)}(\rho _{0},\rho )=\Pr (\lambda _{n}(\widehat{\boldsymbol{%
\theta }}_{c},\widetilde{\boldsymbol{\theta }}_{rc})>\chi
_{1,0.05}^{2}|H_{1})\text{ and }\widehat{\beta }_{n}^{(CLRT)}(\rho _{0},\rho
)=\frac{\dsum\limits_{i=1}^{R}\emph{I}(\lambda _{n}(\widehat{\boldsymbol{%
\theta }}_{c},\widetilde{\boldsymbol{\theta }}_{rc})>\chi _{1,0.05}^{2}|\rho
_{0},\rho )}{R}.
\end{equation*}
Among the composite test-statistics with simulated significance levels
verifying (\ref{a}), at first sight the composite test-statistics with
higher powers should be selected however since in general high powers
correspond to high significance levels, this choice is not straighforward.
For this reason, based on $\beta _{n}^{LRT}-\alpha _{n}^{LRT}$ as baseline,
the efficiencies relative to the composite likelihood ratio test, given by 
\begin{equation*}
e_{n}^{\left( \lambda \right) }=\frac{(\beta _{n}^{(\lambda )}-\alpha
_{n}^{(\lambda )})-(\beta _{n}^{LRT}-\alpha _{n}^{LRT})}{\beta
_{n}^{LRT}-\alpha _{n}^{LRT}},\text{ }\lambda \in \left\{
-1,-0.5,0.2/3,1,1.5\right\} ,
\end{equation*}%
were considered for $n=100,$ $n=200$ and $n=300$. Only the values of the
power for $\lambda =-1/2$ are included in Tables \ref{table3} and \ref%
{table4}, in order to show that the corresponding composite test-statistic
is a good alternative to the composite likelihood ratio test-statistic. The
values of the powers for which the values of $e_{n}^{\left( -1/2\right) }$\
are positive, i.e., the case in which the composite test-statistic
associated to $\lambda =-1/2$ is better that the composite likelihood ratio
test, are shown in bold in Tables \ref{table3} and \ref{table4}. This choice
of $\lambda =-1/2$ divergence based test-statistic has been also recommended
in Morales et al. (1997) and Mart\'{\i}n et al. (2016).

\begin{table}[hbpt]  \tabcolsep2.8pt \small\centering%
$%
\begin{tabular}{cc|cccc|}
\cline{3-6}
&  & $\rho =-0.2$ & $\rho =-0.15$ & $\rho =0$ & $\rho =0.1$ \\ \hline
\multicolumn{1}{|c}{$n=100$} & \multicolumn{1}{|c|}{$CLRT$} & $0.3584$ & $%
0.1604$ & $0.2993$ & $0.7958$ \\ 
\multicolumn{1}{|c}{} & \multicolumn{1}{|c|}{$\lambda =-1/2$} & $\boldsymbol{%
0.3751}$ & $\boldsymbol{0.1750}$ & $\boldsymbol{0.3057}$ & $\boldsymbol{%
0.8076}$ \\ \hline
\multicolumn{1}{|c}{$n=200$} & \multicolumn{1}{|c|}{$CLRT$} & $0.5455$ & $%
0.2227$ & $0.5087$ & $0.9705$ \\ 
\multicolumn{1}{|c}{} & \multicolumn{1}{|c|}{$\lambda =-1/2$} & $\boldsymbol{%
0.5512}$ & $\boldsymbol{0.2322}$ & $\boldsymbol{0.5114}$ & $\boldsymbol{%
0.9737}$ \\ \hline
\multicolumn{1}{|c}{$n=300$} & \multicolumn{1}{|c|}{$CLRT$} & $0.7770$ & $%
0.2705$ & $0.8087$ & $0.9962$ \\ 
\multicolumn{1}{|c}{} & \multicolumn{1}{|c|}{$\lambda =-1/2$} & $\boldsymbol{%
0.7797}$ & $\boldsymbol{0.2795}$ & $\boldsymbol{0.8112}$ & $0.9970$ \\ \hline
\end{tabular}%
$\caption{Simulated powers for $\rho_0=-0.1$.\label{table3}}%
\end{table}%

\begin{table}[hbpt]  \tabcolsep2.8pt \small\centering%
$%
\begin{tabular}{cc|cccc|}
\cline{3-6}
&  & $\rho =0$ & $\rho =0.15$ & $\rho =0.25$ & $\rho =0.3$ \\ \hline
\multicolumn{1}{|c}{$n=100$} & \multicolumn{1}{|c|}{$CLRT$} & $0.8054$ & $%
0.1227$ & $0.1534$ & $0.3689$ \\ 
\multicolumn{1}{|c}{} & \multicolumn{1}{|c|}{$\lambda =-1/2$} & $\boldsymbol{%
0.8118}$ & $\boldsymbol{0.1305}$ & $\boldsymbol{0.1602}$ & $\boldsymbol{%
0.3806}$ \\ \hline
\multicolumn{1}{|c}{$n=200$} & \multicolumn{1}{|c|}{$CLRT$} & $0.9813$ & $%
0.1904$ & $0.2146$ & $0.5818$ \\ 
\multicolumn{1}{|c}{} & \multicolumn{1}{|c|}{$\lambda =-1/2$} & $0.9825$ & $%
0.1920$ & $\boldsymbol{0.2194}$ & $\boldsymbol{0.5957}$ \\ \hline
\multicolumn{1}{|c}{$n=300$} & \multicolumn{1}{|c|}{$CLRT$} & $0.9978$ & $%
0.2591$ & $0.2870$ & $0.7482$ \\ 
\multicolumn{1}{|c}{} & \multicolumn{1}{|c|}{$\lambda =-1/2$} & $0.9979$ & $%
0.2577$ & $\boldsymbol{0.2935}$ & $\boldsymbol{0.7612}$ \\ \hline
\end{tabular}%
$\caption{Simulated powers  for $\rho_0=0.2$.\label{table4}}%
\end{table}%

\section{Conclusions\label{Sec8}}

This paper presents the theoretical background for the development of
statistical tests for testing composite hypotheses when the composite
likelihood is used instead of the classic likelihood of the data. The test
statistic is based on the notion of phi-divergence and its by products, that
is measures of the statistical distance between the theoretical model and
the respective empirical one. The notion of divergence or disparity provides
with abstract methods of estimation and testing and four monographs,
mentioned in the introductory section, developed the state of the art on
this subject.

This work is the first, to the best of our knowledge, which try to link the
notion of composite likelihood with the notion of divergence between
theoretical and empirical models for testing hypotheses. There are several
extensions to this framework which can be considered. The theoretical
framework, presented here, would be extended to develop statistical tests
for testing homogeneity of two or more populations on the basis of composite
likelihood. On the other hand, minimum phi-divergence or disparity
procedures have been observed to provide strong robustness properties in
estimation and testing problems. It would be maybe of interest to proceed in
this direction in a composite likelihood setting.

\section{Appendix\label{Sec9}}

\subsection{Proof of Theorem \protect\ref{Theorem3}\label{ApA}}

Following Sen and Singer (1993, p. 242-3), let $\boldsymbol{\theta }_{n}=%
\boldsymbol{\theta }+n^{-1/2}\boldsymbol{v}$, where $\left\Vert \boldsymbol{v%
}\right\Vert <K^{\ast }$, $0<K^{\ast }<\infty $. Consider now the following
Taylor expansion of the partial derivative of the composite log-density,%
\begin{equation}
\frac{1}{\sqrt{n}}\dsum\limits_{i=1}^{n}\left. \frac{\partial }{\partial 
\boldsymbol{\theta }}c\ell (\boldsymbol{\theta }\mathbf{,}\boldsymbol{y}_{i}%
\mathbf{)}\right\vert _{\boldsymbol{\theta }=\boldsymbol{\theta }_{n}}=\frac{%
1}{\sqrt{n}}\dsum\limits_{i=1}^{n}\frac{\partial }{\partial \boldsymbol{%
\theta }}c\ell (\boldsymbol{\theta }\mathbf{,}\boldsymbol{y}_{i}\mathbf{)+}%
\frac{1}{n}\dsum\limits_{i=1}^{n}\left. \frac{\partial ^{2}}{\partial 
\boldsymbol{\theta }\mathbf{\partial }\boldsymbol{\theta }^{T}}c\ell (%
\boldsymbol{\theta }\mathbf{,}\boldsymbol{y}_{i}\mathbf{)}\right\vert _{%
\boldsymbol{\theta }=\boldsymbol{\theta }_{n}^{\ast }}\sqrt{n}\left( 
\boldsymbol{\theta }_{n}-\boldsymbol{\theta }\right) ,  \label{A0}
\end{equation}%
where $\boldsymbol{\theta }_{n}^{\ast }$ belongs to the line segment joining 
$\boldsymbol{\theta }$ and $\boldsymbol{\theta }_{n}$. Then, observing that
(cf. Theorem 2.3.6 of Sen and Singer, 1993, p. 61)%
\begin{equation*}
\frac{1}{n}\dsum\limits_{i=1}^{n}\frac{\partial ^{2}}{\partial \boldsymbol{%
\theta }\mathbf{\partial }\boldsymbol{\theta }^{T}}c\ell (\boldsymbol{\theta 
}\mathbf{,}\boldsymbol{y}_{i}\mathbf{)}\overset{P}{\underset{n\rightarrow
\infty }{\longrightarrow }}E_{\boldsymbol{\theta }}\left[ \frac{\partial ^{2}%
}{\partial \boldsymbol{\theta }\mathbf{\partial }\boldsymbol{\theta }^{T}}%
c\ell (\boldsymbol{\theta }\mathbf{,}\boldsymbol{Y}\mathbf{)}\right] =E_{%
\boldsymbol{\theta }}\left[ \frac{\partial }{\partial \boldsymbol{\theta }}%
\boldsymbol{u}^{T}(\boldsymbol{\theta }\mathbf{,}\boldsymbol{Y}\mathbf{)}%
\right] =-\boldsymbol{H}(\boldsymbol{\theta }),
\end{equation*}%
equation (\ref{A0}) leads%
\begin{equation}
\left. \frac{1}{\sqrt{n}}\dsum\limits_{i=1}^{n}\frac{\partial }{\partial 
\boldsymbol{\theta }}c\ell (\boldsymbol{\theta }\mathbf{,}\boldsymbol{y}_{i}%
\mathbf{)}\right\vert _{\boldsymbol{\theta }=\boldsymbol{\theta }_{n}}=\frac{%
1}{\sqrt{n}}\dsum\limits_{i=1}^{n}\frac{\partial }{\partial \boldsymbol{%
\theta }}c\ell (\boldsymbol{\theta }\mathbf{,}\boldsymbol{y}_{i}\mathbf{)-}%
\boldsymbol{H}(\boldsymbol{\theta })\sqrt{n}\left( \boldsymbol{\theta }_{n}-%
\boldsymbol{\theta }\right) +o_{P}(1).  \label{A1}
\end{equation}%
Since $\boldsymbol{G}(\boldsymbol{\theta })=\frac{\partial \boldsymbol{g}%
^{T}(\boldsymbol{\theta })}{\partial \boldsymbol{\theta }}$ is continuous in 
$\boldsymbol{\theta }$, it is true that,%
\begin{equation}
\boldsymbol{g}(\boldsymbol{\theta }_{n})=\boldsymbol{G}^{T}(\boldsymbol{%
\theta })\sqrt{n}\left( \boldsymbol{\theta }_{n}-\boldsymbol{\theta }\right)
+o_{P}(1).  \label{A2}
\end{equation}%
Since, the restricted maximum composite likelihood estimator $\widetilde{%
\boldsymbol{\theta }}_{rc}$ should satisfy the likelihood equations%
\begin{align*}
\dsum\limits_{i=1}^{n}\frac{\partial }{\partial \boldsymbol{\theta }}c\ell (%
\boldsymbol{\theta }\mathbf{,}\boldsymbol{y}_{i}\mathbf{)+}\boldsymbol{G}(%
\boldsymbol{\theta })\mathbf{\lambda }& =\boldsymbol{0}_{p}\mathbf{,} \\
\boldsymbol{g}(\boldsymbol{\theta })& =\boldsymbol{0}_{r}\mathbf{,}
\end{align*}%
and in view of (\ref{A1}) and (\ref{A2}) it holds that%
\begin{align*}
\frac{1}{\sqrt{n}}\dsum\limits_{i=1}^{n}\frac{\partial }{\partial 
\boldsymbol{\theta }}c\ell (\boldsymbol{\theta }\mathbf{,}\boldsymbol{y}_{i}%
\mathbf{)-}\boldsymbol{H}(\boldsymbol{\theta })\sqrt{n}\left( \widetilde{%
\boldsymbol{\theta }}_{rc}-\boldsymbol{\theta }\right) +\boldsymbol{G}(%
\boldsymbol{\theta })\frac{1}{\sqrt{n}}\overline{\mathbf{\lambda }}%
_{n}+o_{P}(1)& =\boldsymbol{0}_{p}\mathbf{,} \\
\boldsymbol{G}^{T}(\boldsymbol{\theta })\sqrt{n}(\widetilde{\boldsymbol{%
\theta }}_{rc}-\boldsymbol{\theta })+o_{P}(1)& =\boldsymbol{0}_{p}\mathbf{.}
\end{align*}%
In matrix notation it may be re-expressed as%
\begin{equation*}
\left( 
\begin{array}{cc}
\boldsymbol{H}(\boldsymbol{\theta }) & -\boldsymbol{G}(\boldsymbol{\theta })
\\ 
-\boldsymbol{G}^{T}(\boldsymbol{\theta }) & \boldsymbol{0}_{r\times r}%
\end{array}%
\right) \left( 
\begin{array}{c}
\sqrt{n}(\widetilde{\boldsymbol{\theta }}_{rc}-\boldsymbol{\theta }) \\ 
n^{-1/2}\overline{\mathbf{\lambda }}_{n}%
\end{array}%
\right) =\left( 
\begin{array}{c}
\frac{1}{\sqrt{n}}\dsum\limits_{i=1}^{n}\frac{\partial }{\partial 
\boldsymbol{\theta }}c\ell (\boldsymbol{\theta }\mathbf{,}\boldsymbol{y}_{i}%
\mathbf{)} \\ 
\boldsymbol{0}_{r}%
\end{array}%
\right) +o_{P}(1).
\end{equation*}%
Then 
\begin{equation}
\left( 
\begin{array}{c}
\sqrt{n}(\widetilde{\boldsymbol{\theta }}_{rc}-\boldsymbol{\theta }) \\ 
n^{-1/2}\overline{\mathbf{\lambda }}_{n}%
\end{array}%
\right) =\left( 
\begin{array}{cc}
\boldsymbol{P}(\boldsymbol{\theta }) & \boldsymbol{Q}(\boldsymbol{\theta })
\\ 
\boldsymbol{Q}^{T}(\boldsymbol{\theta }) & \boldsymbol{R}(\boldsymbol{\theta 
})%
\end{array}%
\right) \left( 
\begin{array}{c}
\frac{1}{\sqrt{n}}\dsum\limits_{i=1}^{n}\frac{\partial }{\partial 
\boldsymbol{\theta }}c\ell (\boldsymbol{\theta }\mathbf{,}\boldsymbol{y}_{i}%
\mathbf{)} \\ 
\boldsymbol{0}_{r}%
\end{array}%
\right) +o_{P}(1),  \label{A3}
\end{equation}%
where%
\begin{equation*}
\left( 
\begin{array}{cc}
\boldsymbol{P}(\boldsymbol{\theta }) & \boldsymbol{Q}(\boldsymbol{\theta })
\\ 
\boldsymbol{Q}^{T}(\boldsymbol{\theta }) & \boldsymbol{R}\mathbf{(}%
\boldsymbol{\theta })%
\end{array}%
\right) =\left( 
\begin{array}{cc}
\boldsymbol{H}(\boldsymbol{\theta }) & -\boldsymbol{G}(\boldsymbol{\theta })
\\ 
-\boldsymbol{G}^{T}(\boldsymbol{\theta }) & \boldsymbol{0}_{r\times r}%
\end{array}%
\right) ^{-1}.
\end{equation*}%
This last equation implies (cf. Sen and Singer, 1993, p. 243, eq. (5.6.24)),%
\begin{equation*}
\boldsymbol{P}(\boldsymbol{\theta })=\boldsymbol{H}^{-1}(\boldsymbol{\theta }%
)\left( \boldsymbol{I}_{p}-\boldsymbol{G}(\boldsymbol{\theta })\left( 
\boldsymbol{G}^{T}\mathbf{(}\boldsymbol{\theta }\mathbf{)}\boldsymbol{H}%
^{-1}(\boldsymbol{\theta })\boldsymbol{G}\mathbf{(}\boldsymbol{\theta }%
\mathbf{)}\right) ^{-1}\boldsymbol{G}^{T}(\boldsymbol{\theta })\boldsymbol{H}%
^{-1}(\boldsymbol{\theta })\right) ,
\end{equation*}%
\begin{equation*}
\boldsymbol{Q}(\boldsymbol{\theta })=-\boldsymbol{H}^{-1}(\boldsymbol{\theta 
})\boldsymbol{G}(\boldsymbol{\theta })\left( \boldsymbol{G}^{T}\mathbf{(}%
\boldsymbol{\theta }\mathbf{)}\boldsymbol{H}^{-1}(\boldsymbol{\theta })%
\boldsymbol{G}\mathbf{(}\boldsymbol{\theta }\mathbf{)}\right) ^{-1},
\end{equation*}%
\begin{equation*}
\boldsymbol{R}(\boldsymbol{\theta })=-\left( \boldsymbol{G}^{T}\mathbf{(}%
\boldsymbol{\theta }\mathbf{)}\boldsymbol{H}^{-1}(\boldsymbol{\theta })%
\boldsymbol{G}\mathbf{(}\boldsymbol{\theta }\mathbf{)}\right) ^{-1}.
\end{equation*}%
Based on the central limit theorem (Theorem 3.3.1 of Sen and Singer, 1993,
p. 107) and the Cram\'{e}r-Wold theorem (Theorem 3.2.4 of Sen and Singer,
1993, p. 106) it is obtained%
\begin{equation*}
\frac{1}{\sqrt{n}}\dsum\limits_{i=1}^{n}\frac{\partial }{\partial 
\boldsymbol{\theta }}c\ell (\boldsymbol{\theta }\mathbf{,}\boldsymbol{y}_{i}%
\mathbf{)}\overset{\mathcal{L}}{\underset{n\rightarrow \infty }{%
\longrightarrow }}\mathcal{N}\left( \boldsymbol{0}_{p},Var_{\boldsymbol{%
\theta }}[\boldsymbol{u}(\boldsymbol{\theta }\mathbf{,}\boldsymbol{Y}\mathbf{%
)}]\right) .
\end{equation*}%
with $Var_{\boldsymbol{\theta }}[\boldsymbol{u}(\boldsymbol{\theta }\mathbf{,%
}\boldsymbol{Y}\mathbf{)}]=\boldsymbol{J}\mathbf{(}\boldsymbol{\theta }%
\mathbf{)}$. Then, it follows from (\ref{A3}) that%
\begin{equation*}
\left( 
\begin{array}{c}
\sqrt{n}(\widetilde{\boldsymbol{\theta }}_{rc}-\boldsymbol{\theta }) \\ 
n^{-1/2}\overline{\mathbf{\lambda }}_{n}%
\end{array}%
\right) \overset{\mathcal{L}}{\underset{n\rightarrow \infty }{%
\longrightarrow }}\mathcal{N}\left( \boldsymbol{0},\boldsymbol{\Sigma }%
\right) ,
\end{equation*}%
with%
\begin{equation*}
\boldsymbol{\Sigma }=\left( 
\begin{array}{cc}
\boldsymbol{P}(\boldsymbol{\theta }) & \boldsymbol{Q}(\boldsymbol{\theta })
\\ 
\boldsymbol{Q}^{T}(\boldsymbol{\theta }) & \boldsymbol{R}\mathbf{(}%
\boldsymbol{\theta })%
\end{array}%
\right) \left( 
\begin{array}{cc}
\boldsymbol{J}(\boldsymbol{\theta }) & \boldsymbol{0}_{p\times r} \\ 
\boldsymbol{0}_{r\times p} & \boldsymbol{0}_{r\times r}%
\end{array}%
\right) \left( 
\begin{array}{cc}
\boldsymbol{P}^{T}(\boldsymbol{\theta }) & \boldsymbol{Q}^{T}(\boldsymbol{%
\theta }) \\ 
\boldsymbol{Q}(\boldsymbol{\theta }) & \boldsymbol{R}^{T}\mathbf{(}%
\boldsymbol{\theta })%
\end{array}%
\right) ,
\end{equation*}%
or%
\begin{equation*}
\boldsymbol{\Sigma }=\left( 
\begin{array}{cc}
\boldsymbol{P}(\boldsymbol{\theta })\boldsymbol{J}(\boldsymbol{\theta })%
\boldsymbol{P}^{T}(\boldsymbol{\theta }) & \boldsymbol{P}(\boldsymbol{\theta 
})\boldsymbol{J}(\boldsymbol{\theta })\boldsymbol{Q}^{T}(\boldsymbol{\theta }%
) \\ 
\boldsymbol{Q}^{T}(\boldsymbol{\theta })\boldsymbol{J}(\boldsymbol{\theta })%
\boldsymbol{P}^{T}(\boldsymbol{\theta }) & \boldsymbol{Q}^{T}\mathbf{(}%
\boldsymbol{\theta })\boldsymbol{J}(\boldsymbol{\theta })\boldsymbol{Q}^{T}(%
\boldsymbol{\theta })%
\end{array}%
\right) .
\end{equation*}%
Therefore,%
\begin{equation*}
\sqrt{n}(\widetilde{\boldsymbol{\theta }}_{rc}-\boldsymbol{\theta })\overset{%
\mathcal{L}}{\underset{n\rightarrow \infty }{\longrightarrow }}\mathcal{N}%
\left( \boldsymbol{0}_{p},\boldsymbol{P}(\boldsymbol{\theta })\boldsymbol{J}(%
\boldsymbol{\theta })\boldsymbol{P}^{T}(\boldsymbol{\theta })\right) ,
\end{equation*}%
with%
\begin{align*}
\boldsymbol{P}(\boldsymbol{\theta })& =\boldsymbol{H}^{-1}(\boldsymbol{%
\theta })\left( \boldsymbol{I}_{p}-\boldsymbol{G}(\boldsymbol{\theta }%
)\left( \boldsymbol{G}^{T}\mathbf{(}\boldsymbol{\theta }\mathbf{)}%
\boldsymbol{H}^{-1}(\boldsymbol{\theta })\boldsymbol{G}\mathbf{(}\boldsymbol{%
\theta }\mathbf{)}\right) ^{-1}\boldsymbol{G}^{T}(\boldsymbol{\theta })%
\boldsymbol{H}^{-1}(\boldsymbol{\theta })\right) \\
& =\boldsymbol{H}^{-1}(\boldsymbol{\theta })-\boldsymbol{H}^{-1}(\boldsymbol{%
\theta })\boldsymbol{G}(\boldsymbol{\theta })\left( \boldsymbol{G}^{T}%
\mathbf{(}\boldsymbol{\theta }\mathbf{)}\boldsymbol{H}^{-1}(\boldsymbol{%
\theta })\boldsymbol{G}\mathbf{(}\boldsymbol{\theta }\mathbf{)}\right) ^{-1}%
\boldsymbol{G}^{T}(\boldsymbol{\theta })\boldsymbol{H}^{-1}(\boldsymbol{%
\theta }) \\
& =\boldsymbol{H}^{-1}(\boldsymbol{\theta })+\boldsymbol{Q}(\boldsymbol{%
\theta }\mathbf{)}\boldsymbol{G}^{T}\mathbf{(}\boldsymbol{\theta }\mathbf{)}%
\boldsymbol{H}^{-1}(\boldsymbol{\theta }),
\end{align*}%
and the proof of the lemma is now completed.

\subsection{Proof of Lemma \protect\ref{lemma2} \label{ApB}}

Based on equation (\ref{A3}), above, 
\begin{equation}
\sqrt{n}(\widetilde{\boldsymbol{\theta }}_{rc}-\boldsymbol{\theta })=%
\boldsymbol{P}(\boldsymbol{\theta })\frac{1}{\sqrt{n}}\dsum\limits_{i=1}^{n}%
\frac{\partial }{\partial \boldsymbol{\theta }}c\ell (\boldsymbol{\theta }%
\mathbf{,}\boldsymbol{y}_{i}\mathbf{)}+o_{P}(1).  \label{A4}
\end{equation}%
The Taylor series expansion (\ref{A0}) gives that%
\begin{equation*}
\boldsymbol{0}=\left. \frac{1}{\sqrt{n}}\dsum\limits_{i=1}^{n}\frac{\partial 
}{\partial \boldsymbol{\theta }}c\ell (\boldsymbol{\theta }\mathbf{,}%
\boldsymbol{y}_{i}\mathbf{)}\right\vert _{\boldsymbol{\theta }=\widehat{%
\boldsymbol{\theta }}_{c}}=\frac{1}{\sqrt{n}}\dsum\limits_{i=1}^{n}\frac{%
\partial }{\partial \boldsymbol{\theta }}c\ell (\boldsymbol{\theta }\mathbf{,%
}\boldsymbol{y}_{i}\mathbf{)+}\frac{1}{n}\dsum\limits_{i=1}^{n}\left. \frac{%
\partial ^{2}}{\partial \boldsymbol{\theta }\mathbf{\partial }\boldsymbol{%
\theta }^{T}}c\ell (\boldsymbol{\theta }\mathbf{,}\boldsymbol{y}_{i}\mathbf{)%
}\right\vert _{\boldsymbol{\theta }=\boldsymbol{\theta }_{n}^{\ast }}\sqrt{n}%
(\widehat{\boldsymbol{\theta }}_{c}-\boldsymbol{\theta }),
\end{equation*}%
or%
\begin{equation*}
\frac{1}{\sqrt{n}}\dsum\limits_{i=1}^{n}\frac{\partial }{\partial 
\boldsymbol{\theta }}c\ell (\boldsymbol{\theta }\mathbf{,}\boldsymbol{y}_{i}%
\mathbf{)}=\mathbf{-}\frac{1}{n}\dsum\limits_{i=1}^{n}\left. \frac{\partial
^{2}}{\partial \boldsymbol{\theta }\mathbf{\partial }\boldsymbol{\theta }^{T}%
}c\ell (\boldsymbol{\theta }\mathbf{,}\boldsymbol{y}_{i}\mathbf{)}%
\right\vert _{\boldsymbol{\theta }=\boldsymbol{\theta }_{n}^{\ast }}\sqrt{n}(%
\widehat{\boldsymbol{\theta }}_{c}-\boldsymbol{\theta }),
\end{equation*}%
where $\boldsymbol{\theta }_{n}^{\ast }$ belongs to the line segment joining 
$\boldsymbol{\theta }$ and $\widehat{\boldsymbol{\theta }}_{c}$. Taking into
account Theorem 2.3.6 of Sen and Singer (1993, p. 61), 
\begin{equation*}
\frac{1}{n}\dsum\limits_{i=1}^{n}\left. \frac{\partial ^{2}}{\partial 
\boldsymbol{\theta }\mathbf{\partial }\boldsymbol{\theta }^{T}}c\ell (%
\boldsymbol{\theta }\mathbf{,}\boldsymbol{y}_{i}\mathbf{)}\right\vert _{%
\boldsymbol{\theta }=\boldsymbol{\theta }_{n}^{\ast }}\overset{P}{\underset{%
n\rightarrow \infty }{\longrightarrow }}-\boldsymbol{H}(\boldsymbol{\theta }%
),
\end{equation*}%
and the above two equations lead%
\begin{equation}
\frac{1}{\sqrt{n}}\dsum\limits_{i=1}^{n}\frac{\partial }{\partial 
\boldsymbol{\theta }}c\ell (\boldsymbol{\theta }\mathbf{,}\boldsymbol{y}_{i}%
\mathbf{)}=\boldsymbol{H}(\boldsymbol{\theta })\sqrt{n}(\widehat{\boldsymbol{%
\theta }}_{c}-\boldsymbol{\theta })+o_{P}(1).  \label{A5}
\end{equation}%
Equations (\ref{A4}), (\ref{A5}) and the fact that $\boldsymbol{P}(%
\boldsymbol{\theta }\mathbf{)}=\boldsymbol{H}^{-1}(\boldsymbol{\theta })+%
\boldsymbol{Q}(\boldsymbol{\theta }\mathbf{)}\boldsymbol{G}^{T}\mathbf{(}%
\boldsymbol{\theta }\mathbf{)}\boldsymbol{H}^{-1}(\boldsymbol{\theta })$
give that%
\begin{align*}
\sqrt{n}(\widetilde{\boldsymbol{\theta }}_{rc}-\boldsymbol{\theta })& =%
\boldsymbol{P}(\boldsymbol{\theta })\frac{1}{\sqrt{n}}\dsum\limits_{i=1}^{n}%
\frac{\partial }{\partial \boldsymbol{\theta }}c\ell (\boldsymbol{\theta }%
\mathbf{,}\boldsymbol{y}_{i}\mathbf{)}+o_{P}(1) \\
& =\boldsymbol{P}(\boldsymbol{\theta })\boldsymbol{H}(\boldsymbol{\theta })%
\sqrt{n}(\widehat{\boldsymbol{\theta }}_{c}-\boldsymbol{\theta })+o_{P}(1) \\
& =\left( \boldsymbol{H}^{-1}(\boldsymbol{\theta })+\boldsymbol{Q}(%
\boldsymbol{\theta }\mathbf{)}\boldsymbol{G}^{T}\mathbf{(}\boldsymbol{\theta 
}\mathbf{)}\boldsymbol{H}^{-1}(\boldsymbol{\theta })\right) \boldsymbol{H}(%
\boldsymbol{\theta })\sqrt{n}(\widehat{\boldsymbol{\theta }}_{c}-\boldsymbol{%
\theta })+o_{P}(1),
\end{align*}%
which completes the proof of the lemma.

\subsection{Proof of Theorem \protect\ref{Theorem4} \label{ApC}}

A second order Taylor expansion of $D_{\phi }(\widehat{\boldsymbol{\theta }}%
_{c},\widetilde{\boldsymbol{\theta }}_{rc})$, considered as a function of $%
\widehat{\boldsymbol{\theta }}_{c}$, around $\widetilde{\boldsymbol{\theta }}%
_{rc}$, gives%
\begin{align*}
D_{\phi }(\widehat{\boldsymbol{\theta }}_{c},\widetilde{\boldsymbol{\theta }}%
_{rc})& =D_{\phi }(\widehat{\boldsymbol{\theta }}_{c},\widetilde{\boldsymbol{%
\theta }}_{rc})+\left. \frac{\partial }{\partial \boldsymbol{\theta }}%
D_{\phi }(\boldsymbol{\theta }\mathbf{,}\widetilde{\boldsymbol{\theta }}%
_{rc})\right\vert _{\boldsymbol{\theta }=\widetilde{\boldsymbol{\theta }}%
_{rc}}(\widehat{\boldsymbol{\theta }}_{c}-\widetilde{\boldsymbol{\theta }}%
_{rc}) \\
& +\frac{1}{2}(\widehat{\boldsymbol{\theta }}_{c}-\widetilde{\boldsymbol{%
\theta }}_{rc})^{T}\left. \frac{\partial ^{2}}{\partial \boldsymbol{\theta }%
\mathbf{\partial }\boldsymbol{\theta }^{T}}D_{\phi }(\boldsymbol{\theta }%
\mathbf{,}\widetilde{\boldsymbol{\theta }}_{rc})\right\vert _{\theta =%
\widetilde{\boldsymbol{\theta }}_{rc}}(\widehat{\boldsymbol{\theta }}_{c}-%
\widetilde{\boldsymbol{\theta }}_{rc})+o(\Vert \widehat{\boldsymbol{\theta }}%
_{c}-\widetilde{\boldsymbol{\theta }}_{rc}\Vert ^{2}).
\end{align*}%
Based on Pardo (2006, p. 411-412), we obtain $D_{\phi }(\widetilde{%
\boldsymbol{\theta }}_{rc}\mathbf{,}\widetilde{\boldsymbol{\theta }}_{rc})=0$%
, $\left. \frac{\partial }{\partial \boldsymbol{\theta }}D_{\phi }(%
\boldsymbol{\theta }\mathbf{,}\widetilde{\boldsymbol{\theta }}%
_{rc})\right\vert _{\boldsymbol{\theta }=\widetilde{\boldsymbol{\theta }}%
_{rc}}=\boldsymbol{0}_{p}$ and $\left. \frac{\partial ^{2}}{\partial 
\boldsymbol{\theta }\mathbf{\partial }\boldsymbol{\theta }^{T}}D_{\phi }(%
\boldsymbol{\theta }\mathbf{,}\widetilde{\boldsymbol{\theta }}%
_{rc})\right\vert _{\boldsymbol{\theta }=\widetilde{\boldsymbol{\theta }}%
_{rc}}=\phi ^{\prime \prime }(1)\boldsymbol{J}(\widetilde{\boldsymbol{\theta 
}}_{rc})$. Then, the above equation leads%
\begin{equation*}
\frac{2n}{\phi ^{\prime \prime }(1)}D_{\phi }(\widehat{\boldsymbol{\theta }}%
_{c},\widetilde{\boldsymbol{\theta }}_{rc})=\sqrt{n}(\widehat{\boldsymbol{%
\theta }}_{c}-\widetilde{\boldsymbol{\theta }}_{rc})^{T}\boldsymbol{J}(%
\widetilde{\boldsymbol{\theta }}_{rc})\sqrt{n}(\widehat{\boldsymbol{\theta }}%
_{c}-\widetilde{\boldsymbol{\theta }}_{rc})+no(\Vert \widehat{\boldsymbol{%
\theta }}_{c}-\widetilde{\boldsymbol{\theta }}_{rc}\Vert ^{2}),
\end{equation*}%
or%
\begin{equation}
T_{\phi ,n}(\widehat{\boldsymbol{\theta }}_{c}\mathbf{,}\widetilde{%
\boldsymbol{\theta }}_{rc})=\sqrt{n}(\widehat{\boldsymbol{\theta }}_{c}-%
\widetilde{\boldsymbol{\theta }}_{rc})^{T}\boldsymbol{J}(\widetilde{%
\boldsymbol{\theta }}_{rc})\sqrt{n}(\widehat{\boldsymbol{\theta }}_{c}-%
\widetilde{\boldsymbol{\theta }}_{rc})+no(\Vert \widehat{\boldsymbol{\theta }%
}_{c}-\widetilde{\boldsymbol{\theta }}_{rc}\Vert ^{2}).  \label{A6}
\end{equation}%
On the other hand (cf., Pardo, 2006, p. 63), 
\begin{equation*}
no(||\widehat{\boldsymbol{\theta }}_{c}-\widetilde{\boldsymbol{\theta }}%
_{rc}||^{2})\leq no(\Vert \widehat{\boldsymbol{\theta }}_{c}-\boldsymbol{%
\theta }\Vert ^{2})+no(||\widetilde{\boldsymbol{\theta }}_{rc}-\boldsymbol{%
\theta }||^{2}),
\end{equation*}%
and $no(\Vert \widehat{\boldsymbol{\theta }}_{c}-\boldsymbol{\theta }\Vert
^{2})=o_{P}(1)$, $no(||\widetilde{\boldsymbol{\theta }}_{rc}-\boldsymbol{%
\theta }||^{2})=o_{P}(1)$. Therefore, $o(||\widehat{\boldsymbol{\theta }}%
_{c}-\widetilde{\boldsymbol{\theta }}_{rc}||^{2})=o_{P}(1)$. To apply the
Slutsky's theorem, it remains to obtain the asymptotic distribution of the
quantity 
\begin{equation*}
\sqrt{n}(\widehat{\boldsymbol{\theta }}_{c}-\widetilde{\boldsymbol{\theta }}%
_{rc})^{T}\boldsymbol{J}(\widetilde{\boldsymbol{\theta }}_{rc})\sqrt{n}(%
\widehat{\boldsymbol{\theta }}_{c}-\widetilde{\boldsymbol{\theta }}_{rc}).
\end{equation*}%
From Lemma \ref{lemma2} it is immediately obtained that%
\begin{equation*}
\sqrt{n}(\widehat{\boldsymbol{\theta }}_{c}-\widetilde{\boldsymbol{\theta }}%
_{rc})=\boldsymbol{Q}(\boldsymbol{\theta }\mathbf{)}\boldsymbol{G}^{T}%
\mathbf{(}\boldsymbol{\theta }\mathbf{)}\sqrt{n}(\widehat{\boldsymbol{\theta 
}}_{c}-\boldsymbol{\theta })+o_{P}(1).
\end{equation*}%
On the other hand, we know that%
\begin{equation*}
\sqrt{n}(\widehat{\boldsymbol{\theta }}_{c}-\boldsymbol{\theta })\overset{%
\mathcal{L}}{\underset{n\rightarrow \infty }{\longrightarrow }}\mathcal{N}%
\left( \boldsymbol{0}_{p},\boldsymbol{G}_{\ast }^{-1}(\boldsymbol{\theta }%
\mathbf{)}\right) .
\end{equation*}%
Therefore,%
\begin{equation*}
\sqrt{n}(\widehat{\boldsymbol{\theta }}_{c}-\widetilde{\boldsymbol{\theta }}%
_{rc})\overset{\mathcal{L}}{\underset{n\rightarrow \infty }{\longrightarrow }%
}\mathcal{N}\left( \boldsymbol{0}_{p},\boldsymbol{G}\mathbf{(}\boldsymbol{%
\theta }\mathbf{)}\boldsymbol{Q}^{T}\mathbf{(}\boldsymbol{\theta }\mathbf{%
\mathbf{)}}\boldsymbol{G}_{\ast }^{-1}(\boldsymbol{\theta }\mathbf{)}%
\boldsymbol{Q}(\boldsymbol{\theta }\mathbf{)}\boldsymbol{G}^{T}\mathbf{(}%
\boldsymbol{\theta }\mathbf{)}\right) ,
\end{equation*}%
and taking into account (\ref{A6}) and Corollary 2.1 of Dik and de Gunst
(1985), $T_{\phi ,n}(f_{\widehat{\boldsymbol{\theta }}_{c}},f_{\widetilde{%
\boldsymbol{\theta }}_{rc}})$ converge in law to the random variable $%
\sum_{i=1}^{k}\beta _{i}Z_{i}^{2}$, where $\beta _{i}$, $i=1,...,k$, are the
eigenvalues of the matrix $\boldsymbol{J}(\boldsymbol{\theta })\boldsymbol{G}%
\mathbf{(}\boldsymbol{\theta }\mathbf{)}\boldsymbol{Q}^{T}(\boldsymbol{%
\theta })\boldsymbol{G}_{\ast }^{-1}(\boldsymbol{\theta }\mathbf{)}%
\boldsymbol{Q}(\boldsymbol{\theta })\boldsymbol{G}(\boldsymbol{\theta })^{T}$
and%
\begin{equation*}
k=\mathrm{rank}\left( \boldsymbol{G}\mathbf{(}\boldsymbol{\theta }\mathbf{)}%
\boldsymbol{Q}^{T}\mathbf{(}\boldsymbol{\theta }\mathbf{\mathbf{)}}%
\boldsymbol{G}_{\ast }^{-1}(\boldsymbol{\theta }\mathbf{)}\boldsymbol{Q}(%
\boldsymbol{\theta }\mathbf{)}\boldsymbol{G}^{T}\mathbf{(}\boldsymbol{\theta 
}\mathbf{)}\boldsymbol{J}(\boldsymbol{\theta })\boldsymbol{G}\mathbf{(}%
\boldsymbol{\theta }\mathbf{)}\boldsymbol{Q}^{T}\mathbf{(}\boldsymbol{\theta 
}\mathbf{\mathbf{)}}\boldsymbol{G}_{\ast }^{-1}(\boldsymbol{\theta }\mathbf{)%
}\boldsymbol{Q}(\boldsymbol{\theta }\mathbf{)}\boldsymbol{G}^{T}\mathbf{(}%
\boldsymbol{\theta }\mathbf{)}\right) .
\end{equation*}

\noindent \noindent \textbf{Acknowledgments}. The third author wants to
cordialy thank Prof. Alex de Leon, from the University of Calgary, for
fruitful discussions on composite likelihood methods, some years ago. This
research is partially supported by Grants MTM2012-33740 from Ministerio de
Economia y Competitividad (Spain).

\end{document}